\documentclass[lettersize,hidelinks,journal]{IEEEtran}

\usepackage{cite}
\usepackage{amsmath,amssymb,amsfonts}
\usepackage{graphicx}
\usepackage{textcomp}
\usepackage[dvipsnames]{xcolor}
\usepackage[acronym,shortcuts]{glossaries}
\usepackage{bm}
\usepackage{caption}
\usepackage{subcaption}
\usepackage{relsize}
\usepackage{soul}
\usepackage{algorithm, algpseudocode}
\usepackage{algcompatible}
\usepackage{outlines}
\usepackage{orcidlink}
\usepackage{lipsum,comment}
\usepackage{mathtools}
\usepackage{tabularx,flushend}
\usepackage{amsthm}
\newtheorem{remark}{Remark}
\newtheorem{theorem}{Theorem}
\theoremstyle{definition}
\newtheorem{lemma}{Lemma}

\def\BibTeX{{\rm B\kern-.05em{\sc i\kern-.025em b}\kern-.08em
T\kern-.1667em\lower.7ex\hbox{E}\kern-.125emX}}

\newcommand{\trans}[0]{^{\mathsf{T}}}

\newcommand{\herm}[0]{^{\mathsf{H}}}

\newacronym{RPE}{RPE}{radar parameter estimation}
\newacronym{OTFS}{OTFS}{orthogonal time frequency space}
\newacronym{AFDM}{AFDM}{affine frequency division multiplexing}
\newacronym{MIMO}{MIMO}{multiple-input multiple-output}
\newacronym{SISO}{SISO}{single-input single-output}
\newacronym{ISAC}{ISAC}{integrated sensing and communications}
\newacronym{3D}{3D}{three-dimensional}
\newacronym{2D}{2D}{two-dimensional}
\newacronym{1D}{1D}{one-dimensional}
\newacronym{RX}{RX}{receiver}
\newacronym{TX}{TX}{transmitter}
\newacronym{BF}{BF}{beamforming}
\newacronym{mmWave}{mmWave}{millimeter-wave}
\newacronym{SotA}{SotA}{state-of-the-art}
\newacronym{ULA}{ULA}{uniform linear array}
\newacronym{QAM}{QAM}{quadrature amplitude modulation}
\newacronym{ISFFT}{ISFFT}{inverse symplectic finite Fourier transform}
\newacronym{SFFT}{SFFT}{symplectic finite Fourier transform}
\newacronym{AWGN}{AWGN}{additive white Gaussian noise}
\newacronym{OFDM}{OFDM}{orthogonal frequency division multiplexing}
\newacronym{OCDM}{OCDM}{orthogonal chirp division multiplexing}
\newacronym{BS}{BS}{base station}
\newacronym{UE}{UE}{user equipment}
\newacronym{DFT}{DFT}{discrete Fourier transform}
\newacronym{IDFT}{IDFT}{inverse discrete Fourier transform}
\newacronym{TD}{TD}{time-domain}
\newacronym{wlg}{wlg}{without loss of generality}
\newacronym{CP}{CP}{cyclic prefix}
\newacronym{DAFT}{DAFT}{discrete affine Fourier transform}
\newacronym{IDAFT}{IDAFT}{inverse discrete affine Fourier transform}
\newacronym{CPP}{CPP}{\textit{chirp-periodic} prefix}
\newacronym{IDZT}{IDZT}{inverse discrete Zak transform}
\newacronym{DZT}{DZT}{discrete Zak transform}
\newacronym{ICI}{ICI}{inter-carrier interference}
\newacronym{BER}{BER}{bit error rate}
\newacronym{DoF}{DoF}{degrees-of-freedom}
\newacronym{FD}{FD}{full-duplex}
\newacronym{SIMO}{SIMO}{single-input multiple-output}
\newacronym{MISO}{MISO}{multiple-input single-output}
\newacronym{AoD}{AoD}{angle-of-departure}
\newacronym{AoA}{AoA}{angle-of-arrival}
\newacronym{RF}{RF}{radio frequency}
\newacronym{SIM}{SIM}{stacked intelligent metasurfaces}
\newacronym{FIM}{FIM}{flexible intelligent metasurface}
\newacronym{FPGA}{FPGA}{field programmable gate array}
\newacronym{UPA}{UPA}{uniform planar array}
\newacronym{CC}{CC}{communication-centric}
\newacronym{I/O}{I/O}{input-output}
\newacronym{iid}{i.i.d.}{independent and identically distributed}
\newacronym{IoT}{IoT}{internet of things}
\newacronym{V2X}{V2X}{vehicle-to-everything}
\newacronym{NTN}{NTN}{non-terrestrial network}
\newacronym{LEO}{LEO}{low-earth orbit}
\newacronym{THz}{THz}{terahertz}
\newacronym{EM}{EM}{electromagnetic}
\newacronym{STAR-RIS}{STAR-RIS}{simultaneously transmitting and reflecting reconfigurable intelligent surface}
\newacronym{DoA}{DoA}{direction-of-arrival}
\newacronym{DD}{DD}{doubly-dispersive}
\newacronym{ODDM}{ODDM}{orthogonal delay-Doppler division multiplexing}
\newacronym{LoS}{LoS}{line-of-sight}
\newacronym{NLoS}{NLoS}{non-line-of-sight}
\newacronym{6G}{6G}{sixth generation}
\newacronym{MPDD}{MPDD}{metasurfaces-parameterized DD}
\newacronym{GaBP}{GaBP}{Gaussian Belief Propagation}
\newacronym{MSE}{MSE}{mean-squared-error}
\newacronym{sIC}{soft IC}{soft interference cancellation}
\newacronym{soft RG}{soft RG}{soft replica generation}
\newacronym{BG}{BG}{belief generation}
\newacronym{SGA}{SGA}{scalar Gaussian approximation}
\newacronym{CLT}{CLT}{central limit theorem}
\newacronym{PDF}{PDF}{probability density function}
\newacronym{QPSK}{QPSK}{quadrature phase-shift keying}
\newacronym{LMMSE}{LMMSE}{linear minimum mean square error}
\newacronym{SNR}{SNR}{signal-to-noise ratio}
\newacronym{QoS}{QoS}{quality of service}
\newacronym{CoV}{CoV}{calculus of variations}
\newacronym{CAPA}{CAPA}{continuous aperture array}
\newacronym{FCAPA}{FCAPA}{flexible continuous aperture array}
\newacronym{GL}{GL}{Gauss-Legendre}
\newacronym{DDC MIMO}{DDC MIMO}{DD continuous MIMO}
\newacronym{B5G}{B5G}{beyond fifth generation}
\newacronym{VR}{VR}{virtual reality}
\newacronym{XR}{XR}{extended reality}
\newacronym{ITN}{ITN}{intelligent traffic networks}
\newacronym{SAGIN}{SAGIN}{space-air-ground integrated network}
\newacronym{UAV}{UAV}{unmanned aerial vehicle}
\newacronym{MUSIC}{MUSIC}{Multiple Signal Classification}
\newacronym{ICC}{ICC}{integrated communication and computing}
\newacronym{SINR}{SINR}{signal-to-interference-plus-noise ratio}
\newacronym{WSR}{WSR}{weighted sum rate}
\newacronym{ARPU}{ARPU}{average rate per user}
\newacronym{BCD}{BCD}{block coordinate descent}
\newacronym{PDE}{PDE}{partial differential equation}
\newacronym{EL}{EL}{Euler-Lagrange}
\newacronym{TCA}{TCA}{tightly coupled array}
\newacronym{ELAA}{ELAA}{extremely large-aperture arrays}
\newacronym{LIS}{LIS}{large intelligent surface}
\newacronym{CSI}{CSI}{channel state information}
\newacronym{RIS}{RIS}{reconfigurable intelligent surface}
\newacronym{KKT}{KKT}{Karush-Kuhn-Tucker}
\newacronym{MoM}{MoM}{method of moments}
\newacronym{SVD}{SVD}{singular value decomposition}
\newacronym{RWG}{RWG}{Rao-Wilton-Glisson}
\newacronym{EFIE}{EFIE}{electric field integral equation}
\newacronym{rms}{rms}{root mean square}
\newacronym{EMF}{EMF}{electromagnetic field}
\newacronym{XL-MIMO}{XL-MIMO}{extremely large-scale MIMO}
\newacronym{PD}{PD}{power density}
\newacronym{HS}{HS}{Hilbert-Schmidt}

\begin{document}


\title{Characterization of Continuous Electromagnetic Manifolds via Calculus of Variations}

\author{Kuranage Roche Rayan Ranasinghe\textsuperscript{\orcidlink{0000-0002-6834-8877}},~\IEEEmembership{Graduate Student Member,~IEEE,} \\
Miguel Rodrigo Castellanos\textsuperscript{\orcidlink{0000-0003-1627-1123}},~\IEEEmembership{Member,~IEEE,}
and Giuseppe Thadeu Freitas de Abreu\textsuperscript{\orcidlink{0000-0002-5018-8174}},~\IEEEmembership{Senior Member,~IEEE}
\thanks{K. R. R. Ranasinghe and G. T. F. de Abreu are with the School of Computer Science and Engineering, Constructor University (previously Jacobs University Bremen), Campus Ring 1, 28759 Bremen, Germany (emails: \{kranasinghe, gabreu\}@constructor.university).}
\thanks{Miguel Rodrigo Castellanos is with the Department of Electrical Engineering
and Computer Science, University of Tennessee, Knoxville, TN 37966 USA (email: mrcastelanos@utk.edu).}
\thanks{Parts of this work have been accepted for presentation at the 2026 Asilomar Conference on Signals, Systems, and Computers (IEEE ASILOMAR)~\cite{ranasinghe2026novel}.}
\vspace{-3ex}}



\maketitle

\begin{abstract}
We present a novel \ac{CoV}-based framework for the characterizing of, and beamforming over, continuous electromagnetic manifolds of arbitrary \ac{MIMO} array geometries.
The \ac{EM} manifold, \textit{i.e.}, the set of all physically realizable radiated field vectors parameterized by the array excitation, encodes the full spatial structure of an antenna system, including near-field phase curvature, polarization, and mutual coupling.
Building upon the discrete moment-matrix formulation of the \ac{SotA}, the proposed framework simultaneously overcomes three of its fundamental limitations: \textit{(i)}~the point-source approximation error incurred by the near-field radiation operator; \textit{(ii)}~the confinement of the beamforming space to the $N$-dimensional subspace dictated by the hardware port count; and \textit{(iii)}~the generalization to arbitrary array geometries.
To this end, each mesh element is modeled as a two-dimensional planar patch whose spatially averaged Green's function is evaluated via \ac{GL} quadrature, yielding a strictly more accurate near-field representation at negligible additional cost, while a continuous feeding function $w(\mathbf{p})\!\in\! L^2(\mathcal{S}_\mathrm{T})$, introduced as the infinite-dimensional limit of the $N$-port network, lifts the optimization onto a hardware-decoupled current subspace of dimension $K\!\gg\! N$.
As an application example, we employ the proposed \ac{CoV}-based framework to derive closed-form optimal beamformers for both unconstrained field-strength maximization, and a near-field pattern synthesis under a \ac{PD} and region constraints, establishing their exact analogy to the discrete and generalized matched filters.
Full-wave MATLAB Antenna Toolbox validation confirms consistent near-field accuracy gains over the \ac{SotA} baseline for both linear and planar geometries at comparable computational cost, while a spectral analysis of the steering operator quantifies the additional spatial \acp{DoF} unlocked by the continuous model, which the region-constrained beamformer is shown to exploit for markedly sharper spatial suppression.
\end{abstract}


\glsresetall

\vspace{-3ex}
\section{Introduction}

\IEEEPARstart{M}{ulti-antenna} technologies have been central to the evolution of wireless systems, exploiting an ever-growing number of spatial \acp{DoF}~\cite{PoonTIT2005}.
For the past two decades, the dominant modeling paradigm, has been far-field models based on plane-wave assumptions. This has been adequate because array apertures were small relative to communication distances, mutual coupling was manageable, and beamforming objectives were confined to azimuth/elevation steering.
Three converging trends now collectively invalidate these assumptions and necessitate an electromagnetically-consistent treatment.
\emph{First}, the transition to \acf{XL-MIMO} and \acf{CAPA}~\cite{WangTWC2025-CAPAOptimal,zhang2023pattern}, part of a broader movement toward holographic \ac{MIMO} and large intelligent surfaces that treat the aperture as a quasi-continuous radiating structure~\cite{Huang_WC_2020,Hu_TSP_2018,Gong_COMST_2024}, dramatically expands the Rayleigh distance $d_R = 2D^2/\lambda$, where $D$ is the aperture diameter.
For example, with a $1\,\mathrm{m}$ aperture at $5\,\mathrm{GHz}$, $d_R \approx 333\,\mathrm{m}$, placing typical users firmly in the near field, where spherical wavefront curvature, spatial non-stationarity, and polarization mixing can no longer be neglected~\cite{wang2025analytical}.
Comprehensive treatments of this regime~\cite{Lu_COMST_2024,Wang_WC_2024,Bjornson_OJCOMS_2020} establish that, once the aperture spans many wavelengths, far-field conditions such as uniform plane waves and distance-invariant array responses break down, so that both channel estimation~\cite{Cui_TWC_2023} and beamforming~\cite{Bjornson_Asilomar_2021} must be recast under a spherical-wave model.
\emph{Second}, sub-wavelength element spacing introduces significant mutual coupling that the scalar model ignores entirely; the embedded element currents obtained from a full-wave solver implicitly encode all coupling effects, and any model that bypasses this physics systematically misrepresents the available \acp{DoF}~\cite{Wallace2004}.
Coupling is not merely a parasitic effect to be de-embedded~\cite{Gupta_TAP_1983}.
A circuit-theoretic account of the antenna-channel interface~\cite{Ivrlac_TCAS_2010} shows that the currents a surface can actually support, and hence the superdirective and higher-order radiation modes it can excite~\cite{Marzetta_Asilomar_2019,GustafssonTAP2013}, are inseparable from the coupling, which is precisely the information that a manifold built on embedded currents preserves.
\emph{Third}, near-field sensing and \ac{ISAC} applications require precise localization of scatterers at distances where the radial components of the Green's function are non-negligible~\cite{elbir2024near}; errors in the radiation operator translate directly into localization bias.
Indeed, the same spherical-wave structure that complicates communication is what enables near-field \emph{beam focusing} at a point rather than merely in a direction~\cite{Zhang_TWC_2022_beamfocus,Nepa_APM_2017}, underpinning location-based services, wideband near-field beamforming~\cite{Cui_TWC_2024_wideband}, and high-resolution source localization~\cite{Friedlander_TSP_2019_nf}; in every case, an inaccurate near-field radiation operator maps directly into a bias in the estimated position or focal spot.
Together, these trends mark a qualitative shift in operating regime: the far-field scalar model is no longer a mild approximation, but a structural mismatch with the underlying physics, motivating the electromagnetically rigorous framework developed herein.

The foundational framework for \ac{EM} manifold characterization of discrete antenna arrays was established in~\cite{CastellanosTWC2025}, where the array response is parameterized via a moment matrix derived from a point-source (Dirac delta) approximation of each mesh segment, and beamforming is formulated over a finite $N$-port feed weight vector.
This line of work builds upon a rich history of manifold-based array characterization~\cite{YangTAP2019,FriedlanderTSP2020}, which captures near-field phase structure, polarization, and coupling that the classical plane-wave manifold discards.
At a more fundamental level, the number of independent field distributions an aperture can radiate is finite and dictated by its electrical size -- a principle formalized by electromagnetic information theory and by the classical results on the degrees of freedom of scattered and band-limited fields~\cite{Bucci_TAP_1989,Bucci_TAP_1998,Migliore_TAP_2008,Jensen_TAP_2008}. This viewpoint has re-emerged in the study of large intelligent surfaces and their fundamental limits~\cite{Dardari_JSAC_2020,Decarli_TWC_2021}, and it directly motivates the spectral (effective-rank) characterization of the steering operator developed in Section~\ref{sec:spectral}.
While powerful, this approach carries three limitations: \textit{(i)} the point-source radiation operator introduces approximation errors in the near-field, \textit{(ii)} the model is confined to $N$-dimensional current subspaces imposed by the hardware port count, and \textit{(iii)} the generalization to arbitrary array geometries is not made explicit.

To address these limitations simultaneously, it is essential to adopt a \emph{continuous} \ac{EM} framework.
Unlike discrete optimization models, which either face a prohibitive computational burden or resort to suboptimal approximations of continuous functional programs, a continuous \ac{CoV} approach fundamentally decouples the optimization space from the hardware architecture~\cite{WangTWC2025-CAPAMU,WangTWC2025-CAPAOptimal,WangTWC2026-CAPAMIMO}.
By modeling the mesh elements as realistic \ac{2D} planar patches evaluated via \ac{GL} quadrature, the radiation operator attains superior near-field accuracy without incurring prohibitive cost, while introducing a continuous feeding function as the infinite-dimensional limit of the $N$-port network lifts the beamforming-subspace restriction, maximizing the exploitable \acp{DoF} and seamlessly supporting the arbitrary planar topologies~\cite{YuanAWPL2022,JelinekTAP2017} essential to modern \ac{MIMO}.
Continuous and holographic aperture models have previously been approached through Fourier plane-wave and wavenumber-domain expansions~\cite{Pizzo_JSAC_2020,Pizzo_TWC_2022,Sanguinetti_TWC_2023,Wei_JSTSP_2022}, and realized in an approximate, discretized form by reconfigurable and holographic metasurfaces~\cite{DiRenzo_JSAC_2020,Wu_TWC_2019,Deng_TVT_2021}. The present \ac{CoV} treatment is complementary to these efforts: rather than fixing a basis or a hardware template a priori, it yields the exact functional optimum over $L^2(\mathcal{S}_\mathrm{T})$, thereby furnishing a hardware-agnostic performance bound against which such architectures can be benchmarked.

This paper addresses all three limitations within a unified \ac{CoV} framework.
The specific contributions are as follows.
\begin{itemize}
    \item \textbf{Patch-based radiation operator:} Each mesh element is modeled as a realistic \ac{2D} planar patch rather than a point source. 
    The spatially averaged Green's function over each patch is evaluated via tensor-product \ac{GL} quadrature, yielding a strictly more accurate near-field representation at negligible additional cost.
    \item \textbf{Continuous feeding framework:} We introduce a continuous feeding function $w(\mathbf{p}) \in L^2(\mathcal{S}_\mathrm{T})$ as the infinite-dimensional limit of the $N$-port network, enabling optimization over a $K$-dimensional current subspace ($K \gg N$) that is decoupled from hardware port constraints.
    \item \textbf{Region-constrained beamforming:} Leveraging the continuous framework, we derive a closed-form optimal beamformer for field-strength maximization subject to a power density constraint over an extended spatial exclusion region $\mathcal{P}_\mathrm{con}$, establishing its continuous analogue to the discrete 
    generalized matched filter.
    \item \textbf{Support for arbitrary \ac{2D} geometries:} Unlike~\cite{CastellanosTWC2025}, the proposed model supports arbitrary planar transmit surface geometries (and extends naturally to volumetric ones), validated for both linear and planar arrays.
\end{itemize}

\vspace{-1.5ex}
\subsection{Organization and Notation}
\vspace{-1ex}

\textit{Organization:} The rest of this paper is organized as follows.
Section~\ref{sec:system_model} develops the continuous \ac{EM} system model, progressing from a continuous radiating surface with a discrete $N$-port feed to its infinite-dimensional limit with a continuous feeding function, and establishes the relationship of the resulting operators to the classical \ac{MoM}/\ac{RWG} formulation.
Section~\ref{subsec:numerical_implementation} details the patch-based numerical implementation, the point-source baseline of the \ac{SotA}, the validation against full-wave simulations, and the spectral characterization of the steering operator.
Section~\ref{sec:bf_cov} formulates and solves, in closed form via the \ac{CoV}, the field-strength maximization and region-constrained \ac{PD} beamforming problems.
Finally, Section~\ref{sec:conclusion} concludes the paper, with the main proofs deferred to the appendices.

\textit{Notation:} All scalars are represented by upper or lowercase letters, while column vectors and matrices are denoted by bold lowercase and uppercase letters, respectively.
The diagonal matrix constructed from vector $\mathbf{a}$ is denoted by diag($\mathbf{a}$), while $\mathbf{A}\trans$, $\mathbf{A}\herm$, $\mathbf{A}^{1/2}$, and $\mathbf{A}(i,j)$ denote the transpose, Hermitian, square root and the $(i,j)$-th element of a matrix $\mathbf{A}$, respectively.
The convolution and Kronecker product are respectively denoted by $*$ and $\otimes$, while $\mathbf{I}_N$ and $\mathbf{F}_N$ represent the $N\times N$ identity and the normalized $N$-point \ac{DFT} matrices, respectively.
The sinc function is expressed as $\text{sinc}(a) \triangleq \frac{\sin(\pi a)}{\pi a}$, and $j\triangleq\sqrt{-1}$ denotes the elementary complex number.
The Dirac delta function is denoted by $\delta(\cdot)$.
The Lebesgue measure of a Euclidean subspace $\mathcal{S}$ is denoted by $|\mathcal{S}|$.
The absolute value and Euclidean norm are denoted by $|\cdot|$ and $||\cdot||$, respectively.

\vspace{-1ex}
\section{System Model}
\label{sec:system_model}
\subsection{Continuous Surface with Discrete Feeding Network}
\label{subsec:linear_model}
\vspace{-0.5ex}

Let $\mathbf{s} = [s_x, s_y, s_z]\trans \in \mathcal{S}_\mathrm{T}$ denote an arbitrary point (analogous to a single discrete Hertzian dipole) on a continuous transmit surface $\mathcal{S}_\mathrm{T}$, and similarly, $\mathbf{r} = [r_x, r_y, r_z]\trans \in \mathcal{S}_\mathrm{R}$ denote an arbitrary point on a continuous receive surface $\mathcal{S}_\mathrm{R}$.
Hereafter, it will be generally assumed that $\mathcal{S}_\mathrm{T}$ and $\mathcal{S}_\mathrm{R}$ are embedded in a homogeneous medium (e.g. free space).

Next, let $\bm{j}(\mathbf{s}) \in \mathbb{C}^{3\times 1}$ be an elementary current density at a point $\mathbf{s}$ on $\mathcal{S}_\mathrm{T}$\footnote{To be more explicit, $\bm{j}(\mathbf{s},\omega)$ denotes the Fourier transform of the current density at the point $\mathbf{s}$, where $\omega = 2\pi f / c = 2\pi / \lambda$ denotes the angular frequency, $f$ is the signal frequency, and $\lambda$ is the signal wavelength.
However, due to the narrowband single-carrier assumption, we omit the current density's explicit dependence on $\omega$, hereafter denoting it $\bm{j}(\mathbf{s})$.}, such that the electric field produced by a current density distribution over the transmit surface $\mathcal{S}_\mathrm{T}$, at point $\mathbf{r}$ on the receive surface can be described by \cite{PoonTIT2005}
\begin{equation}
\label{eq:E_integral}
\mathbf{e}(\mathbf{r}) = \int_{\mathcal{S}_\mathrm{T}} \mathbf{G}(\mathbf{r}, \mathbf{s}) \bm{j}(\mathbf{s}) \, {\rm d}\mathbf{s} \in \mathbb{C}^{3\times 1},
\end{equation}
where $\mathbf{G}(\mathbf{r}, \mathbf{s}) \in \mathbb{C}^{3\times 3}$ represents the Green's function \cite{YuanAWPL2022},
given by
\begin{equation}
\label{eq:Greens_function}
\mathbf{G}(\mathbf{r}, \mathbf{s}) = \left( \mathbf{I}_3 + \frac{\nabla \nabla}{\kappa^2} \right) \frac{e^{-j \kappa \|\mathbf{r} - \mathbf{s}\|}}{4\pi\|\mathbf{r} - \mathbf{s}\|},
\end{equation}
with $\kappa = \omega/c$ denoting the wavenumber and $\nabla$ denoting the vector differential operator in the three-dimensional Cartesian coordinate system.

Here, $\mathbf{I}_3 \in \mathbb{R}^{3\times 3}$ is the identity matrix and $\nabla\nabla$ is the dyadic (outer) gradient operator, whose $(i,j)$-th entry is $\partial^2/\partial r_i\,\partial r_j$ applied to the scalar Green's function $g(\mathbf{r},\mathbf{s}) = e^{-j\kappa\|\mathbf{r}-\mathbf{s}\|}/(4\pi\|\mathbf{r}-\mathbf{s}\|)$.
The operator $\mathbf{I}_3 + \nabla\nabla/\kappa^2$ thus encodes both the transverse and the longitudinal field components, the latter of which is negligible in the far field but dominant in the reactive near field.
Crucially, $\mathbf{G}(\mathbf{r},\mathbf{s})$ is smooth for $\mathbf{r} \neq \mathbf{s}$ but singular on the diagonal $\mathbf{r} = \mathbf{s}$, a property that motivates the spatial averaging of the radiation operator introduced in Section~\ref{subsec:numerical_implementation}.

\begin{remark}
    Equation \eqref{eq:E_integral} can be used to derive many optimization problems aimed at designing the current densities $\bm{j}(\mathbf{s})$ in order to achieve specific goals \cite{WangTWC2025-CAPAMU,WangTWC2025-CAPAOptimal,WangTWC2026-CAPAMIMO}.
\end{remark}

A practical mechanism to generate a specific current density distribution on the transmit surface $\mathcal{S}_\mathrm{T}$ is to excite it simultaneously via several feed points.
In particular, let $\mathbf{w} = [w_1,\dots,w_n,\dots,w_N]\trans \in \mathbb{C}^{N \times 1}$ denote the weights of an $N$-port feed into an arbitrary surface $\mathcal{S}_\mathrm{T}$.
Then, the current density induced at each point $\mathbf{s}$ on $\mathcal{S}_\mathrm{T}$ due to these feeds can be expressed as
\begin{equation}
    \label{eq:feed_induced_J}
    \bm{j}(\mathbf{s}; \mathbf{w}) = \sum_{n=1}^N w_n \bm{j}_n(\mathbf{s}) \in \mathbb{C}^{3\times 1},
\end{equation}
where each $\bm{j}_n(\mathbf{s})$ is an elementary current density due to the excitation of the $n$-th port\footnote{This is the continuous analog of the ``moment matrix'' in \cite{CastellanosTWC2025}.} with sinusoidal current of 1 Ampere rms, and we have slightly abused our notation to explicitly emphasize the dependence of the superimposed current density $\bm{j}(\mathbf{s}; \mathbf{w})$ on the weights $\mathbf{w}$.

Substituting equation \eqref{eq:feed_induced_J} into equation \eqref{eq:E_integral} readily yields
\begin{align}
\label{eq:E_integral_induced_J}
\mathbf{e}(\mathbf{r}, \mathbf{w}) &= \sum_{n=1}^N w_n \underbrace{ \int_{\mathcal{S}_\mathrm{T}} \mathbf{G}(\mathbf{r}, \mathbf{s})  \bm{j}_n(\mathbf{s}) \, {\rm d}\mathbf{s} }_{\triangleq \mathbf{a}_n(\mathbf{r}) \in \mathbb{C}^{3 \times 1}}  \in \mathbb{C}^{3\times 1} \nonumber \\[-2ex]
&= \sum_{n=1}^N w_n \mathbf{a}_n(\mathbf{r}) = \mathbf{A}(\mathbf{r}) \mathbf{w},
\end{align}
where we implicitly defined the continuous steering vectors $\mathbf{a}_n(\mathbf{r})$ and the corresponding continuous steering matrix $\mathbf{A}(\mathbf{r}) \triangleq [\mathbf{a}_1(\mathbf{r}), \dots, \mathbf{a}_n(\mathbf{r}), \dots, \mathbf{a}_N(\mathbf{r})] \in \mathbb{C}^{3 \times N}$.

Equation \eqref{eq:E_integral_induced_J} can now be recognized as analogous to  equation (24) of \cite{CastellanosTWC2025} under \emph{continuous} \ac{EM} currents and \emph{discrete} feeds. The continuous steering vectors are analogous to the embedded antenna patterns of the array elements, and the continuous steering matrix represents what is known as the analytic manifold in antenna theory \cite{FriedlanderTSP2020,YangTAP2019}.
To be clear, unlike \cite[Eq.(24)]{CastellanosTWC2025}, equation \eqref{eq:E_integral_induced_J} does not result from or imply the approximation of the continuous radiating structure onto a discrete equivalent with basis on Herzian dipoles.
Instead only excitation feeds needed to generate $\bm{j}(\mathbf{s}; \mathbf{w})$ are discrete, or better, independent, which in turn can be exploited to formulate and solve optimization problems over the sets $\mathbf{w}\in\mathbb{C}^{N\times 1}$.
We further emphasize, that as a consequence of the continuous analysis described above, the model is not reliant on the accuracy of a discrete approximation as in \cite{CastellanosTWC2025}.

\subsection{Relationship to the MoM and RWG Basis Functions}
\label{subsec:mom_rwg}

The \acf{MoM} with \acf{RWG} basis functions~\cite{Makarov2001} is the canonical numerical technique for solving the \acf{EFIE} for the unknown surface currents on a radiating structure.
In the \ac{MoM}/\ac{RWG} formulation, the surface current is expanded as $\bm{j}(\mathbf{s}) = \sum_{k} I_k \boldsymbol{\Lambda}_k(\mathbf{s})$, where the $\boldsymbol{\Lambda}_k$ are \ac{RWG} functions defined on triangular mesh-element pairs, and the coefficients $\{I_k\}$ are obtained by enforcing the boundary conditions via Galerkin testing.
The resulting impedance matrix encodes the mutual couplings through dyadic Green's-function integrals, including singular self-terms that demand specialized quadrature.

The framework developed here differs fundamentally in purpose.
Rather than solving for the unknown currents, our goal is to \emph{characterize the radiated-field manifold} for a prescribed set of port excitations, using the physical embedded currents $\bm{j}_n(\mathbf{s})$ determined by a full-wave solver as an input.
The radiation operator introduced in Section~\ref{subsec:numerical_implementation} then maps these currents to the radiated field at observation points $\mathbf{r} \notin \mathcal{S}_\mathrm{T}$, where the integrand is smooth and standard quadrature applies without any singularity treatment.
The contribution here is therefore not an alternative current solver, but a higher-fidelity \emph{post-processing} operator that replaces the zeroth-order centroid rule of~\cite{CastellanosTWC2025} with spatially averaged patch operators, improving near-field accuracy without revisiting the underlying \ac{MoM} solve.

A related body of work characterizes antenna arrays directly from their far-field patterns~\cite{PoonTIT2005}, or via scalar channel models that absorb mutual coupling into a fixed correction matrix~\cite{Wallace2004}.
While computationally convenient, these approaches conflate the physical current distribution with its far-field projection, discarding the near-field phase structure.
The Hertzian-dipole approximation of~\cite{CastellanosTWC2025} partially recovers this near-field information by retaining the full dyadic Green's function, but evaluates it at the mesh centroids only -- a zeroth-order rule whose error grows as the observation point approaches the surface, and which the present work systematically corrects.

\subsection{Continuous Surface with Continuous Feeding Function}
\label{subsec:cont_feeding_model}

The elegance of \eqref{eq:E_integral_induced_J} deserves some reflection.
In particular, the expression suggests that sophisticated radiation patterns can be designed over continuous surfaces based on the superposition of a discrete number of elementary current densities, in \emph{likeness to} (but certainly not exactly as) a Fourier series\footnote{We emphasize that, unlike Fourier Series, the elementary current densities $\bm{j}_n(\mathbf{s})$ are not harmonically related.}.
However, in order to obtain sufficient control over the current density distributions that can be designed over the transmit surface $\mathcal{S}_\mathrm{T}$, one or both of the following conditions must be met: $i$) the elementary current densities $\bm{j}_n(\mathbf{s})$ themselves must be able to exploit all the \acp{DoF} on the surface; and/or $ii$)
the number of ports $N$ need to be sufficiently large.

Condition $i$ is somewhat undesirable and can be set aside for the time being, since it translates to the requirement that a suitable basis of elementary current densities be designed, which can be potentially addressed in a follow up work.
For this reason, we focus hereafter on condition $ii$, compounded with the assumption that $\bm{j}_n(\mathbf{s})$ are sinusoidal components of the same frequency $\omega$.

This motivates us to consider the limiting case when $N\to\infty$, where the discrete feeding network discussed above tends to a continuous feeding function.
This motivation is also backed by \ac{SotA} on optimal designs for the construction of such feeds as done in \cite{GustafssonTAP2013,JelinekTAP2017}.
Then, let $w(\mathbf{p}) \in L^2(\mathcal{S}_\mathrm{T})$ denote a continuous feeding excitation function at a point $\mathbf{p}$ on the transmitting surface $\mathcal{S}_\mathrm{T}$.
Then, the induced current density can be modeled by the superposition
\begin{equation}
    \label{eq:feed_induced_J_cont}
    \bm{j}(\mathbf{s}) = \int_{\mathcal{S}_\mathrm{T}} w(\mathbf{p}) \bm{j}(\mathbf{s}, \mathbf{p}) \, {\rm d}\mathbf{p} \in \mathbb{C}^{3\times 1},
\end{equation}
where $\bm{j}(\mathbf{s}, \mathbf{p}) \in \mathbb{C}^{3\times 1}$ denotes the embedded element current response, i.e., the current density induced at surface location $\mathbf{s}$ when the aperture is excited by a 1 Ampere rms feed applied at location $\mathbf{p}$.

\begin{remark}[Convergence of Discrete to Continuous Feeding]
The continuous integral in~\eqref{eq:feed_induced_J_cont} is the formal $L^2$ limit of the discrete superposition in~\eqref{eq:feed_induced_J}. 
Specifically, for a sequence of uniform port grids with spacing $\Delta p \to 0$ and weights $w_n = w(\mathbf{p}_n)\Delta p$, the Riemann sums converge in $L^2(\mathcal{S}_\mathrm{T})$ to the integral in~\eqref{eq:feed_induced_J_cont}, provided $w \in L^2(\mathcal{S}_\mathrm{T})$ and $\bm{j}(\mathbf{s},\mathbf{p})$ is square-integrable jointly in $(\mathbf{s},\mathbf{p})$.
\end{remark}

\begin{remark}[Continuous Feeding as a Theoretical Aperture Bound]
\label{rem:capa_bound}
The continuous feeding function $w(\mathbf{p}) \in L^2(\mathcal{S}_T)$ should be interpreted as a
\emph{theoretical upper bound} on the beamforming performance achievable from the aperture
$\mathcal{S}_T$, rather than as a directly realizable hardware architecture.
In the numerical evaluations of Section~\ref{subsec:PD_BF}, this bound is operationalized
by granting independent complex amplitude and phase control to each of the $K$ mesh segments
(see Section~\ref{para:num_cont_feeding}), yielding a $K$-dimensional control space.
Any physical $N$-port feed network ($N < K$) confines realizable current distributions to
an $N$-dimensional subspace thereof; the continuous model thus subsumes all $N$-port
realizations and provides a principled, hardware-agnostic upper bound.
The rate at which this bound is approached as the control dimension increases is governed by the spectral
decay of the steering operator $\mathcal{A}$, characterized in Section~\ref{sec:spectral},
and is illustrated numerically in Fig.~\ref{fig:dof_scaling}.
\end{remark}

\begin{remark}[Physical Realizability of the Continuous Feed]
\label{rem:feed_realizability}
The continuous feeding model~\eqref{eq:feed_induced_J_cont} is presented as a theoretical construct that defines an upper bound on the achievable beamforming performance (cf. Remark~\ref{rem:capa_bound}); its physical realization is, however, non-trivial.
In standard antenna modeling, each feed port imposes a localized boundary condition, typically a delta-gap voltage source or a coaxial probe, that uniquely determines the current distribution over the entire surface through the governing integral equation.
Under such a model, specifying $w(\mathbf{p})$ as a continuous function over $\mathcal{S}_\mathrm{T}$ would over-constrain the system: once the discrete port excitations are fixed, the surface current is already determined by the full-wave boundary-value problem, and $w(\mathbf{p})$ cannot be prescribed independently.
The transition from~\eqref{eq:feed_induced_J} to~\eqref{eq:feed_induced_J_cont} is therefore best understood as a mathematical limiting argument rather than a prescription for a directly realizable feed architecture.
The design of structured $N$-port networks whose aggregate response approximates a target $w(\mathbf{p})$, for instance, along the lines of the optimal-current syntheses in~\cite{GustafssonTAP2013,JelinekTAP2017}, is a promising direction for future work.
\end{remark}

Therefore, $\bm{j}(\mathbf{s}, \mathbf{p})$ in \eqref{eq:feed_induced_J_cont} is the continuously induced analogue of the discretely embedded current densities $\bm{j}_n(\mathbf{s})$ in \eqref{eq:feed_induced_J}, used in the discrete feeding network case\footnote{This interpretation aligns with the continuous moment-method framework used in array manifold characterization and is the natural infinite-dimensional extension of the ``moment matrix'' representation in \cite{CastellanosTWC2025}.}.

Then, substituting equation \eqref{eq:feed_induced_J_cont} into equation \eqref{eq:E_integral} yields
\begin{align}
\label{eq:E_integral_induced_J_cont}
\mathbf{e}(\mathbf{r}) &= \int_{\mathcal{S}_\mathrm{T}} w(\mathbf{p}) \underbrace{ \int_{\mathcal{S}_\mathrm{T}} \mathbf{G}(\mathbf{r}, \mathbf{s}) \bm{j}(\mathbf{s}, \mathbf{p}) \, {\rm d}\mathbf{s} }_{\triangleq \mathbf{a}(\mathbf{r}, \mathbf{p}) \in \mathbb{C}^{3 \times 1}} \, {\rm d}\mathbf{p} \in \mathbb{C}^{3\times 1} \nonumber \\
&= \int_{\mathcal{S}_\mathrm{T}} w(\mathbf{p}) \mathbf{a}(\mathbf{r}, \mathbf{p}) \, {\rm d}\mathbf{p},
\end{align}
where we implicitly defined the continuous steering vector $\mathbf{a}(\mathbf{r}, \mathbf{p})$.

\begin{remark}[Regularity of the Continuous Steering Vector]
\label{rem:regularity}
For observation points $\mathbf{r} \notin \mathcal{S}_\mathrm{T}$, the Green's function $\mathbf{G}(\mathbf{r},\mathbf{s})$ is smooth in $\mathbf{s}$, and $\mathbf{a}(\mathbf{r},\mathbf{p})$, as a function of $\mathbf{p}$, lies in $L^2(\mathcal{S}_\mathrm{T})$ componentwise, inheriting the square-integrability of $\bm{j}(\mathbf{s},\mathbf{p})$.
This regularity is essential for the well-posedness of the \ac{CoV} problems in Section \ref{sec:bf_cov}, and ensures that the \ac{GL} quadrature approximations of the integrals in Theorems~\ref{theorem_1} and~\ref{thm:PD_opt} converge exponentially fast in the quadrature order $N_q$.
\end{remark}

Equation \eqref{eq:E_integral_induced_J_cont} can then be identified as an equivalent of \cite[Eq. (24)]{CastellanosTWC2025} under \emph{continuous} \ac{EM} currents and \emph{continuous} feeding, which can be used to formulate and solve optimization problems.

This alternative representation gives us a way to apply the \ac{CoV} to formulate and solve optimization problems with respect to the \emph{feeds} as well as the current densities.
For example, formulating a field strength maximization problem using equation \eqref{eq:E_integral_induced_J_cont} would yield a low-complexity closed-form which can be approximated well via the \ac{GL} quadrature.

Since some optimization problems formulated hereafter requires the explicit model after receive polarization, let $\mathbf{u}_r \in \mathbb{R}^{3 \times 1}$ denote the polarization direction of the receiver.
Then, the effective electric field captured at the receiver can be expressed as
\begin{equation}
\label{eq:scalar_e_field}
    e(\mathbf{r}) = \mathbf{u}_r\herm \mathbf{e}(\mathbf{r}) = \int_{\mathcal{S}_\mathrm{T}} w(\mathbf{p}) a(\mathbf{r}, \mathbf{p}) \, {\rm d}\mathbf{p} \in \mathbb{C}.
\end{equation}
where we explicitly define the scalar steering vector to be $a(\mathbf{r},\mathbf{p}) \triangleq \mathbf{u}_r\trans \mathbf{a}(\mathbf{r}, \mathbf{p})$.

\section{Model Validation via Numerical Implementation \& Evaluation}
\label{subsec:numerical_implementation}

\subsection{Numerical Implementation}

The continuous field representation in \eqref{eq:E_integral_induced_J_cont} defines a linear operator mapping surface current densities to radiated fields. For numerical evaluation, we adopt a segment-based discretization
of the transmit surface, resulting in a finite-dimensional linear operator consistent with the implemented solver.

\subsubsection{Proposed Method}

\paragraph{Segment-Based Current Representation}

The transmit surface $\mathcal{S}_\mathrm{T}$ is partitioned into $K$
non-overlapping mesh elements with centroids
$\{\mathbf{s}_k \in \mathbb{R}^{3 \times 1}\}_{k=1}^K$ and areas $\{A_k\}_{k=1}^K$.
Each element is represented by a localized vector current moment.

The induced surface current density is approximated as a superposition of
short oriented dipole segments
\begin{equation}
\mathbf{j}(\mathbf{s};\mathbf{w}) \approx \sum_{k=1}^{K} \mathbf{m}_k(\mathbf{w})
\,\delta_\ell(\mathbf{s}-\mathbf{s}_k) \in \mathbb{C}^{3 \times 1},
\label{eq:segment_current_model_twc}
\end{equation}
where $\mathbf{m}_k(\mathbf{w}) \in \mathbb{C}^{3\times1}$ denotes the Cartesian current coefficient associated with element $k$, and $\delta_\ell$ represents a finite-length distribution supported along a local segment centered at $\mathbf{s}_k$.

Stacking the segment coefficients yields
\begin{equation}
\mathbf{m}(\mathbf{w})
=
[\mathbf{m}_1\trans,\dots,\mathbf{m}_K\trans]\trans
\in \mathbb{C}^{3K\times1}.
\end{equation}

\paragraph{Embedded Current Response Matrix}

Let the structure be excited by $N$ feed ports.
For each unit port excitation, the electromagnetic solver provides the
corresponding segment current vector.
Stacking these responses defines the embedded current matrix
\begin{equation}
\label{eq:M_disc}
\mathbf{M} \in \mathbb{C}^{3K\times N}.
\end{equation}

For arbitrary complex feed weights $\mathbf{w}\in\mathbb{C}^{N\times1}$, the induced segment currents are given by
\begin{equation}
\mathbf{m}(\mathbf{w}) = \mathbf{M}\mathbf{w}.
\label{eq:embedded_map_twc}
\end{equation}

\paragraph{Discretized Radiation Operator}

Rather than approximating the mesh elements as \ac{1D} point sources or wire segments, we model each element as a realistic \ac{2D} planar patch. 
We define an effective patch dimension 
\begin{equation}
L_k = \sqrt{A_k}.
\end{equation}

For each patch $k$, we can then define a local tangent plane spanned by two orthogonal unit vectors: $\hat{\mathbf{d}}_{k,1}$, representing the dominant embedded current flow (obtained from the mean embedded current direction across port excitations, as formalized in~\eqref{eq:normal_d}), and $\hat{\mathbf{d}}_{k,2}$, the orthogonal tangent vector on the surface.
Concretely, $\hat{\mathbf{d}}_{k,2}$ is the unit vector in the tangent plane of $\mathcal{S}_\mathrm{T}$ at $\mathbf{s}_k$ satisfying $\hat{\mathbf{d}}_{k,2} \perp \hat{\mathbf{d}}_{k,1}$ and $\hat{\mathbf{d}}_{k,2} \perp \hat{\mathbf{n}}_k$, where $\hat{\mathbf{n}}_k$ is the outward surface normal at $\mathbf{s}_k$; this construction adapts automatically to arbitrary linear and planar surface geometries, requiring no modification of the radiation operator.

Then, the averaged radiated field contribution of patch $k$ at an observation point $\mathbf{r} \notin \mathcal{S}_\mathrm{T}$ is obtained by integrating the Green's function over the local patch area. 
Applying a change of variables to map the local patch coordinates $(x,y) \in [-L_k/2, L_k/2]^2$ to the standard reference domain $(\xi,\eta) \in [-1,1]^2$, the integral is given by
\begin{equation}
\mathbf{K}_k(\mathbf{r})
\approx
\frac{1}{4} \int_{-1}^{1} \int_{-1}^{1}
\mathbf{G}\!\left(
\mathbf{r},
\mathbf{s}_k + \frac{L_k}{2} \big(\xi \hat{\mathbf{d}}_{k,1} + \eta \hat{\mathbf{d}}_{k,2}\big)
\right)
d\xi \, d\eta,
\label{eq:patch_integral_twc}
\end{equation}
where the factor of $1/4$ arises from the Jacobian of the spatial transformation normalized by the patch area $A_k = L_k^2$, yielding the spatially averaged field operator per unit current moment.

Since the integrand is smooth in the near- and far-field regions away from the source, the \ac{2D} integral is evaluated efficiently via a tensor-product Gauss-Legendre quadrature given by
\begin{align}
\label{eq:GL_quadrature_2D}
\mathbf{K}_k(\mathbf{r}) &\approx  \\
&\hspace{-5ex}\frac{1}{4} \sum_{q_1=1}^{N_q} \sum_{q_2=1}^{N_q} \omega_{q_1} \omega_{q_2} \, \mathbf{G}\left(\mathbf{r}, \mathbf{s}_k + \frac{L_k}{2} \big(\xi_{q_1} \hat{\mathbf{d}}_{k,1} + \eta_{q_2} \hat{\mathbf{d}}_{k,2}\big)\right), \nonumber
\end{align}
where $\{\xi_q, \omega_q\}_{q=1}^{N_q}$ are the standard \ac{1D} Gauss-Legendre nodes and weights on the interval $[-1,1]$.

In practice, because the observation point $\mathbf{r}$ is strictly separated from the source patch (avoiding the singularity of the Green's function), the integrand is highly smooth. 
Consequently, the Gauss–Legendre quadrature converges exceedingly fast. 
For standard sub-wavelength mesh elements (e.g., $A_k \ll \lambda^2$), empirical evaluations demonstrate that a low quadrature order of $N_q = 2$ (yielding only four evaluation points per patch) is sufficient to achieve high-fidelity convergence. 
This ensures that the proposed patch-based continuous formulation maintains a low computational complexity profile while providing a strictly superior physical representation compared to zeroth-order point-source approximations.






\paragraph{Matrix Formulation}

Define the radiation matrix
\begin{equation}
\mathbf{K}(\mathbf{r})
=
[\mathbf{K}_1(\mathbf{r}),\dots,\mathbf{K}_K(\mathbf{r})]
\in \mathbb{C}^{3\times3K},
\end{equation}
where each block
$\mathbf{K}_k(\mathbf{r}) \in \mathbb{C}^{3\times3}$
maps the local vector current $\mathbf{m}_k$ to its field contribution.

The total radiated field is then
\begin{equation}
\mathbf{e}(\mathbf{r})
=
\mathbf{K}(\mathbf{r})\mathbf{m}(\mathbf{w})
=
\mathbf{K}(\mathbf{r})\mathbf{M}\mathbf{w}.
\label{eq:final_matrix_field_twc}
\end{equation}

For a receiver with polarization vector $\mathbf{u}_r$,
the scalar received field is
\begin{equation}
e(\mathbf{r})
=
\mathbf{u}_r\trans
\mathbf{K}(\mathbf{r})
\mathbf{M}
\mathbf{w}.
\end{equation}

\paragraph{Numerical Representation of the Continuous Feeding Function} 
\label{para:num_cont_feeding}

The embedded current response matrix $\mathbf{M} \in \mathbb{C}^{3K \times N}$ defined in~\eqref{eq:M_disc} rigorously models the physical bottleneck of a discrete feeding network: regardless of aperture size, the current distribution is confined to an $N$-dimensional subspace. 
In the experiments of Section~\ref{subsec:numerical_implementation}, $N = 16$ hardware ports yield $K = 1{,}152$ mesh segments, so $K/N = 72$, providing a $72\times$ expansion in the \emph{control-space} dimension under the continuous model.
We stress that this is an expansion of the number of independently tunable excitations, not of the physically radiating degrees of freedom: the latter are bounded by the effective rank of the steering operator (Section~\ref{sec:spectral}), which is markedly smaller than $K$ but, crucially, still much larger than $N$.

To evaluate the theoretical performance of $w(\mathbf{p}) \in L^2(\mathcal{S}_\mathrm{T})$, we approximate the infinite-dimensional Hilbert space by granting independent amplitude and phase control over each of the $K$ mesh segments. 
The $k$-th segment excitation is constrained to act along the dominant current direction $\hat{\mathbf{d}}_{k,1}$, defined as the unit vector aligned with the mean real part of the embedded current 
$\mathbf{M}_{(k)} \triangleq \mathbf{M}(3k{-}2:3k,\,:)$ across all port excitations
\begin{equation}
\label{eq:normal_d}
    \hat{\mathbf{d}}_{k,1} = 
    \frac{\sum_{n=1}^{N} \mathrm{Re}\{\mathbf{M}_{(k)}\mathbf{e}_n\}}
         {\big\|\sum_{n=1}^{N} \mathrm{Re}\{\mathbf{M}_{(k)}\mathbf{e}_n\}\big\|},
\end{equation}
where $\mathbf{e}_n$ is the $n$-th standard basis vector.

Let $\mathbf{w}_c \in \mathbb{C}^{K \times 1}$ denote the discretized continuous feeding function evaluated at the segment centroids. 
We can then define the continuous spatial mapping matrix $\mathbf{M}_c \in \mathbb{R}^{3K \times K}$ as a sparse block matrix, where the $k$-th column maps the scalar excitation $w_{c,k} \in \mathbb{C}$ strictly to the local dominant embedded current direction $\hat{\mathbf{d}}_{k,1} \in \mathbb{R}^{3 \times 1}$ of the $k$-th segment. 

The modeling choice in \eqref{eq:normal_d} ensures that $\mathbf{M}_c$ is physically consistent with the EM structure of the aperture; heuristically, exciting a segment against its natural current direction would demand reactive power the surface cannot readily support. This intuition is formalized by the dominant-direction restriction~\eqref{eq:normal_d}, whose approximation error is negligible for the arrays considered (Section~\ref{sec:spectral}).
If the mean current vanishes (e.g., for symmetric geometries), $\hat{\mathbf{d}}_{k,1}$ is taken as the dominant left singular vector of $\mathrm{Re}\{\mathbf{M}_{(k)}\}$.

\begin{remark}[Per-Segment Polarization Degrees of Freedom]
\label{rem:pol_dof}
The restriction of each segment to the single direction $\hat{\mathbf{d}}_{k,1}$ in~\eqref{eq:normal_d} and one complex excitation per segment is exact only when the segment's embedded current is \emph{rank-one in direction} across port excitations, i.e., when the element radiates an essentially fixed polarization. This is made precise by the singular value decomposition of the segment block $\mathbf{M}_{(k)} = \mathbf{U}_k\boldsymbol{\Sigma}_k\mathbf{V}_k^{\mathsf{H}}$, whose left singular vectors $\mathbf{U}_k = [\mathbf{u}_{k,1},\mathbf{u}_{k,2},\mathbf{u}_{k,3}]$ are the \emph{excitation-independent} spatial current directions supported by the element and $\sigma_{k,1}\!\ge\!\sigma_{k,2}\!\ge\!\sigma_{k,3}\!\ge\!0$ their strengths. The single-direction model sets $\hat{\mathbf{d}}_{k,1}\!=\!\mathbf{u}_{k,1}$ and is tight when the ratio $\sigma_{k,2}/\sigma_{k,1}$, which is precisely the per-segment polarization leakage, vanishes. For the dipole and bowtie elements considered the ratio sits at the numerical floor, so the restriction is essentially exact here.

This is not universal, however. In dual-polarized, multi-mode, and \emph{polarization-reconfigurable} elements the dominant current direction is \emph{deliberately} made excitation-dependent, so that changing the feed rotates the radiated polarization~\cite{Christodoulou_ProcIEEE_2012,Rodrigo_TAP_2014}; there $\sigma_{k,2}/\sigma_{k,1}=\mathcal{O}(1)$ and a single fixed direction is inadequate. The framework accommodates such elements without modification by assigning segment $k$ its \emph{polarization rank}
\begin{equation}
d_k \triangleq \bigl|\{\, i : \sigma_{k,i}/\sigma_{k,1}\ge \epsilon_{\mathrm{p}} \,\}\bigr| \in \{1,2,3\},
\label{eq:pol_rank}
\end{equation}
for a small threshold $\epsilon_{\mathrm{p}}>0$, and replacing the $k$-th column of $\mathbf{M}_c$ by the $3\times d_k$ block $[\mathbf{u}_{k,1},\dots,\mathbf{u}_{k,d_k}]$ driven by $d_k$ independent complex weights. 

The control dimension then grows from $K$ to $\sum_{k}d_k$, and the present formulation~\eqref{eq:normal_d} is recovered as the special case $d_k\!\equiv\!1$. The vectors $\{\mathbf{u}_{k,i}\}$ are the local analogue of the structure's characteristic current modes~\cite{Cabedo-Fabres_APM_2007}, so $d_k$ is simply the number of independent polarizations the $k$-th element can excite.
\end{remark}

Consequently, the induced segment currents under the theoretical continuous feeding paradigm can be modeled as
\begin{equation}
    \mathbf{m}(\mathbf{w}_c) = \mathbf{M}_c \mathbf{w}_c \in \mathbb{C}^{3K \times 1}.
    \label{eq:m_cont}
\end{equation}

Substituting~\eqref{eq:m_cont} into~\eqref{eq:final_matrix_field_twc} yields the discretized received field
\begin{equation}
    e(\mathbf{r}) = \mathbf{u}_r\trans \mathbf{K}(\mathbf{r}) \mathbf{M}_c \mathbf{w}_c.
    \label{eq:e_cont}
\end{equation}

By utilizing $\mathbf{M}_c$ instead of $\mathbf{M}$, the proposed \ac{CoV} numerical formulation exploits $K$ spatial degrees of freedom. 
This allows the optimization framework to synthesize the highly complex, spatially continuous current distributions required to solve strict regional pattern synthesis problems that finite $N$-port arrays physically cannot satisfy.

\paragraph{Approximation Properties and Convergence}

The segment-based approximation introduces two sources of error: the patch quadrature error in $\mathbf{K}_k(\mathbf{r})$, and the finite-$K$ discretization error in representing $w(\mathbf{p})$. 
For the former, since the Green's function is smooth for $\mathbf{r} \notin \mathcal{S}_\mathrm{T}$, the \ac{GL} quadrature with $N_q = 2$ achieves relative errors below $10^{-3}$ for $A_k \leq \lambda^2/100$, as confirmed 
empirically in Fig.~\ref{fig:linear_arrays_all}. 
For the latter, as $K \to \infty$ with $A_k \to 0$, the piecewise-constant approximation of $w(\mathbf{p})$ converges in $L^2(\mathcal{S}_\mathrm{T})$ to any target feeding function, so the continuous performance bounds are approached as $K$ increases, and are observed numerically to be approached monotonically from below once $K$ is sufficiently large that the segment size falls below a fraction of a wavelength.
In practice, $K \gg N$ is sufficient; the choice $K = 1{,}152$, $N = 16$ used here provides results indistinguishable from higher-$K$ evaluations for the scenarios considered.

\subsubsection{Baseline method for Comparison} 
The method from \cite{CastellanosTWC2025} applies the same approximation as in \eqref{eq:segment_current_model_twc} but takes $\delta_\ell$ to be a Dirac delta function. In essence, each segment is approximated as a point source rather than a short dipole with finite length. The radiated field contribution is then computed as
\begin{equation}
\mathbf{K}_k(\mathbf{r})
= 
\mathbf{G}\!\left(
\mathbf{r},
\mathbf{s}_k \right),
\label{eq:segment_integral_twc}
\end{equation} 
and the field is obtained from \eqref{eq:final_matrix_field_twc}. 

We note that, in contrast to the results here, the validation simulations in \cite{CastellanosTWC2025} approximate $A_k$ as being the same for each segment to account for the scenario in which the mesh characteristics are unknown.

\subsection{Numerical Evaluation for the Electric Field}

We validate the proposed model by comparing the electric field approximations against electromagnetic simulations performed using the MATLAB Antenna Toolbox for both linear and planar arrays. 
Two array element configurations are considered: a half-wave dipole array and a bowtie triangular array, all tuned to $5$~GHz. 
The effective array moment matrix $\mathbf{M}$ is extracted from MATLAB by measuring the current distribution under single-element excitation, with the MoM segmentation applied automatically.

We compare the simulated field $e_\mathrm{sim}(\mathbf{r})$ generated using EHfields against the proposed model 
$e(\mathbf{r})$ from~\eqref{eq:e_cont} using the relative error
\begin{equation}
    \text{Relative error} = \frac{\|e_\mathrm{sim}(\mathbf{r}) - e(\mathbf{r})\|}{\|e_\mathrm{sim}(\mathbf{r})\|},
    \label{eq:rel_error}
\end{equation}
evaluated as a function of distance from the array. 

\begin{figure}[t!]
    \centering
    \begin{subfigure}{0.92\columnwidth}
        \centering
        \includegraphics[width=\columnwidth]{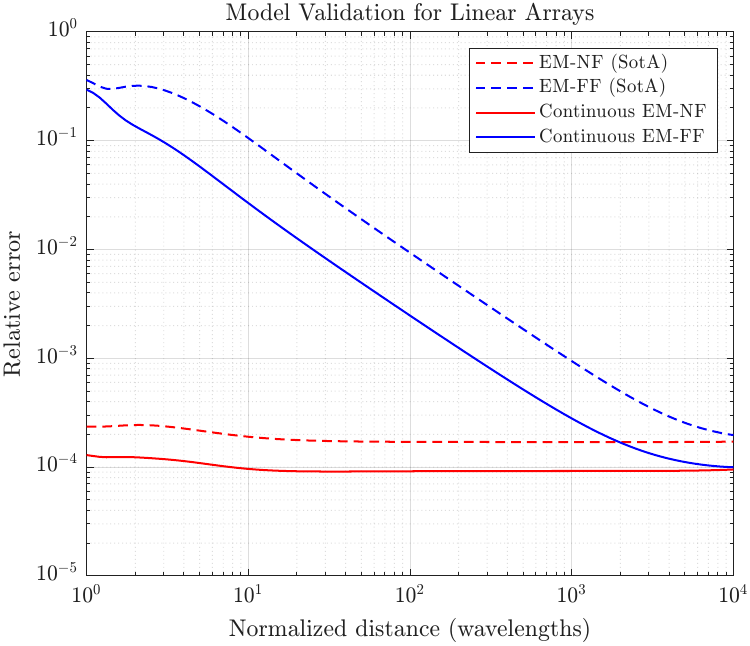}
        \vspace{-3ex}
        \caption{An 8 element Dipole array with $\lambda/2$ spacing.}
        \label{fig:linear_dipole_default}
    \end{subfigure}
    \\
    \begin{subfigure}{0.92\columnwidth}
        \centering
        \includegraphics[width=\columnwidth]{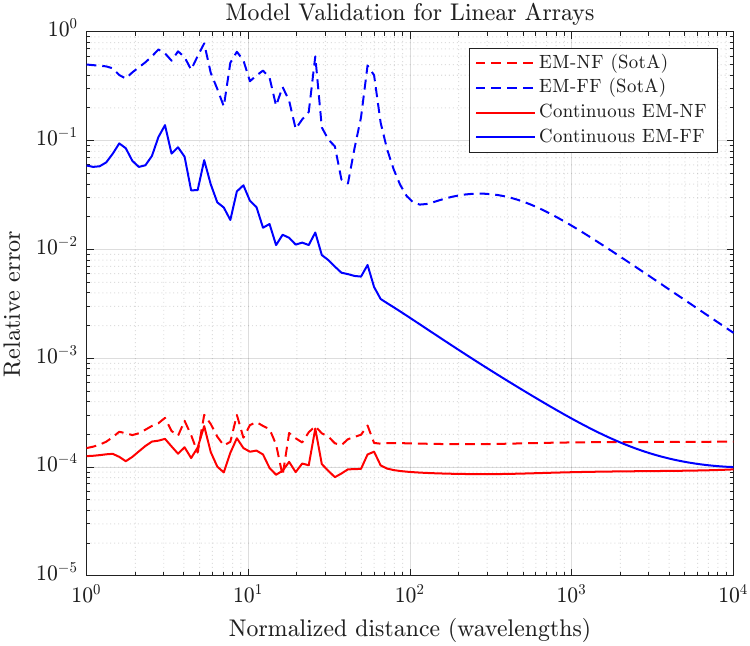}
        \vspace{-3ex}
        \caption{An 8 element Dipole array with $4\lambda$ spacing.}
        \label{fig:linear_diploe_lambda4}
    \end{subfigure}
    \\
    \begin{subfigure}{0.92\columnwidth}
        \centering
        \includegraphics[width=\columnwidth]{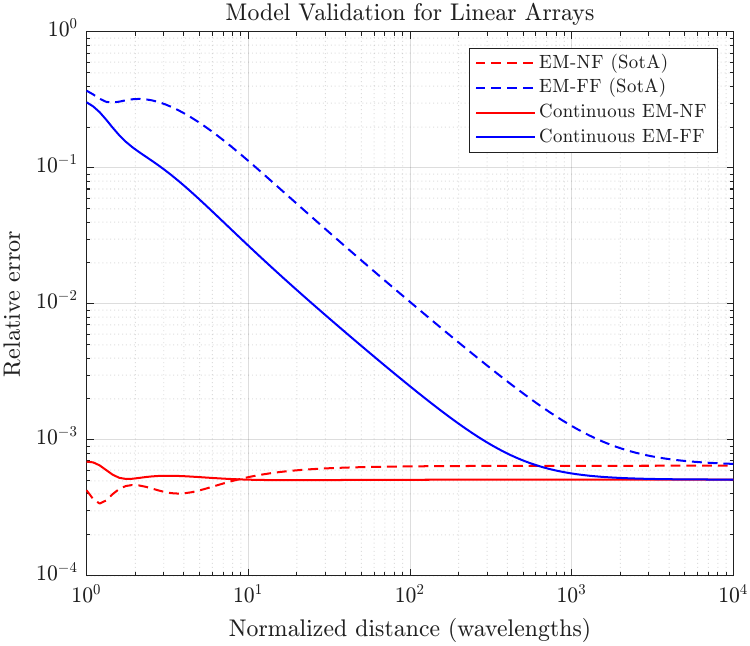}
        \vspace{-3ex}
        \caption{An 8 element bowtie triangular array with $\lambda/2$ spacing.}
        \label{fig:linear_bowtie}
    \end{subfigure}
    \caption{Model validation for linear arrays with an azimuth angle of 120$^\circ$ and an elevation angle of 30$^\circ$.}
    \label{fig:linear_arrays_all}
\end{figure}

\begin{figure}[t!]
    \centering
    \begin{subfigure}{0.92\columnwidth}
        \centering
        \includegraphics[width=\columnwidth]{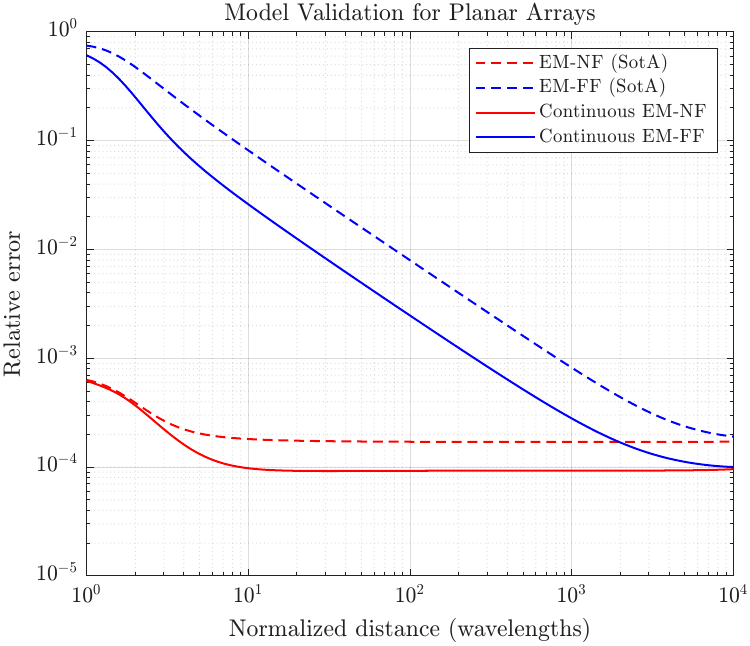}
        \vspace{-3ex}
        \caption{An 2x4 element Dipole array with $\lambda/2$ spacing.}
        \label{fig:planar_dipole_default}
    \end{subfigure}
    \\
    \begin{subfigure}{0.92\columnwidth}
        \centering
        \includegraphics[width=\columnwidth]{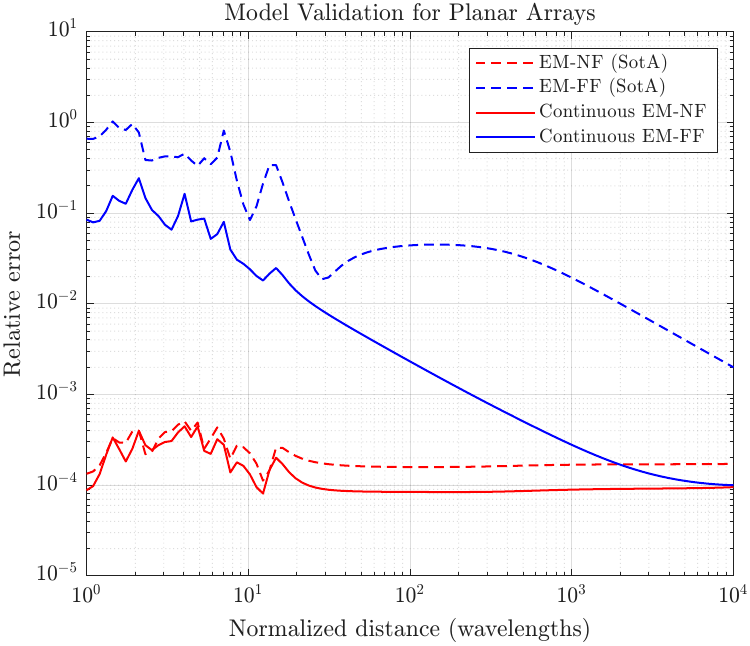}
        \vspace{-3ex}
        \caption{An 2x4 element Dipole array with $4\lambda$ spacing.}
        \label{fig:planar_diploe_lambda4}
    \end{subfigure}
    \\
    \begin{subfigure}{0.92\columnwidth}
        \centering
        \includegraphics[width=\columnwidth]{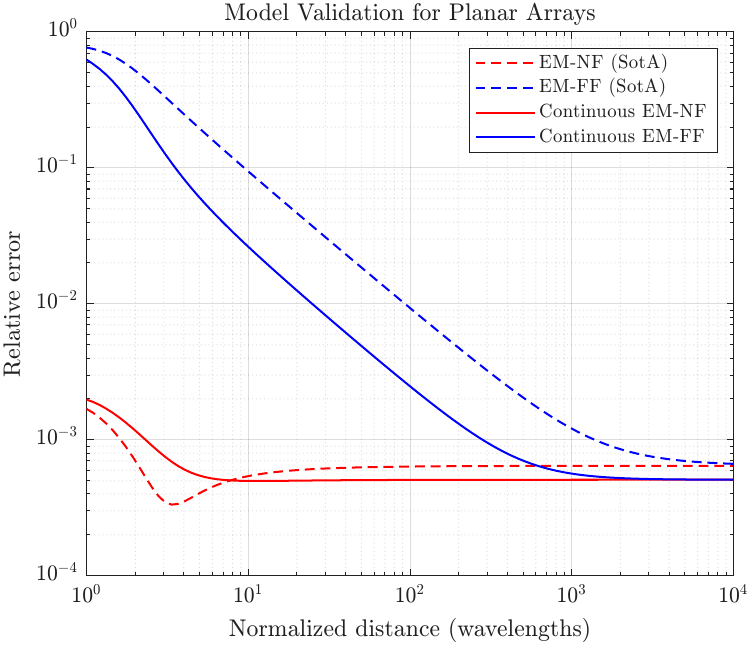}
        \vspace{-3ex}
        \caption{An 2x4 element bowtie triangular array with $\lambda/2$ spacing.}
        \label{fig:planar_bowtie}
    \end{subfigure}
    \caption{Model validation for planar arrays with an azimuth angle of 120$^\circ$ and an elevation angle of 30$^\circ$.}
    \label{fig:planar_arrays_all}
\end{figure}

As a baseline, we use the MoM method originally proposed in \cite{CastellanosTWC2025} and summarized in \eqref{eq:segment_integral_twc}.

Fig. \ref{fig:linear_arrays_all} shows that the proposed near-field model achieves low relative error across all distances for the all the array variations, with the far-field approximation converging to it as expected. 
While there are some larger deviations when the element spacing increases as seen from Fig. \ref{fig:linear_diploe_lambda4}, the model accuracies remain consistent throughout.
Next, Fig. \ref{fig:linear_bowtie} shows results with bowtie triangular elements, where the same gain in model accuracy is seen.

Next, Fig \ref{fig:planar_arrays_all} showcases the same set of results but now for a \ac{2D} planar array. 
As seen from the figure, the gains in relative error remain consistent throughout all the described scenarios.

To confirm that the near-field accuracy gain is a property of the operator rather than an artifact of the single observation direction used above, Fig.~\ref{fig:acc_angle} sweeps the observation azimuth over $[0^\circ,180^\circ]$ at a fixed elevation of $30^\circ$. The proposed patch operator attains a lower relative error than the point-source baseline across the entire angular range, at both $r=1.5\lambda$ and $r=5\lambda$. The advantage is thus uniform in angle, not specific to the $120^\circ$ cut of Figs.~\ref{fig:linear_arrays_all}--\ref{fig:planar_arrays_all}, and it persists with distance rather than vanishing, consistent with the distance sweeps therein.

Fig.~\ref{fig:quad_conv} justifies the choice $N_q=2$: measured against a high-order ($N_q=16$) self-reference, the tensor-product \ac{GL} quadrature converges super-algebraically, its relative error falling below $10^{-6}$ already at $N_q=2$ (four points per patch) and reaching machine precision by $N_q=4$, while the per-patch cost grows only as $N_q^2$. The operator therefore sits at the knee of this curve, securing quadrature-exact near-field fidelity at a fixed four-fold cost over the zeroth-order baseline, independent of array size.

\begin{figure}[t!]
    \centering
    \includegraphics[width=\columnwidth]{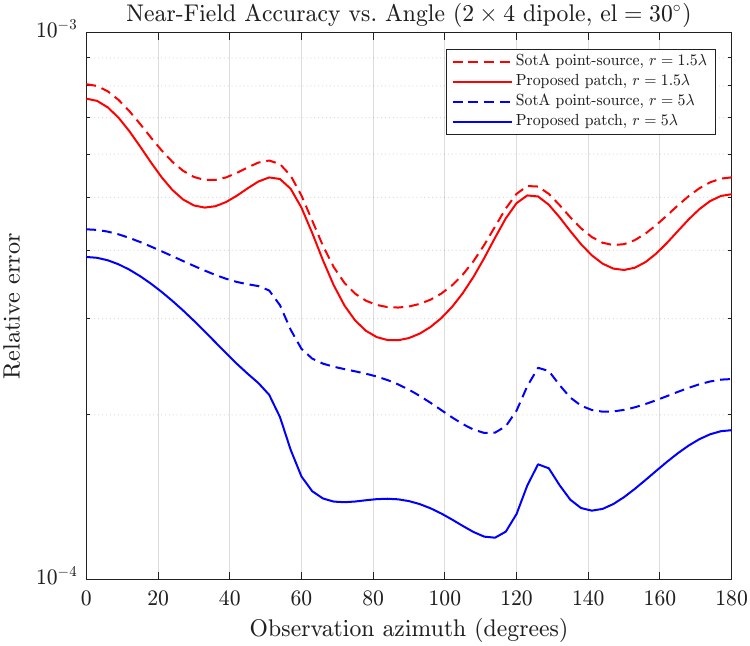}
    \caption{Relative field error versus observation azimuth (elevation $30^\circ$) for the $2\times4$ dipole array (the validated configuration of Fig.~\ref{fig:planar_arrays_all}). The proposed patch operator is uniformly more accurate than the point-source baseline across all angles, at both $r=1.5\lambda$ and $r=5\lambda$, with the advantage widening at the larger distance.}
    \vspace{-3ex}
    \label{fig:acc_angle}
\end{figure}

\begin{figure}[t!]
    \centering
    \includegraphics[width=\columnwidth]{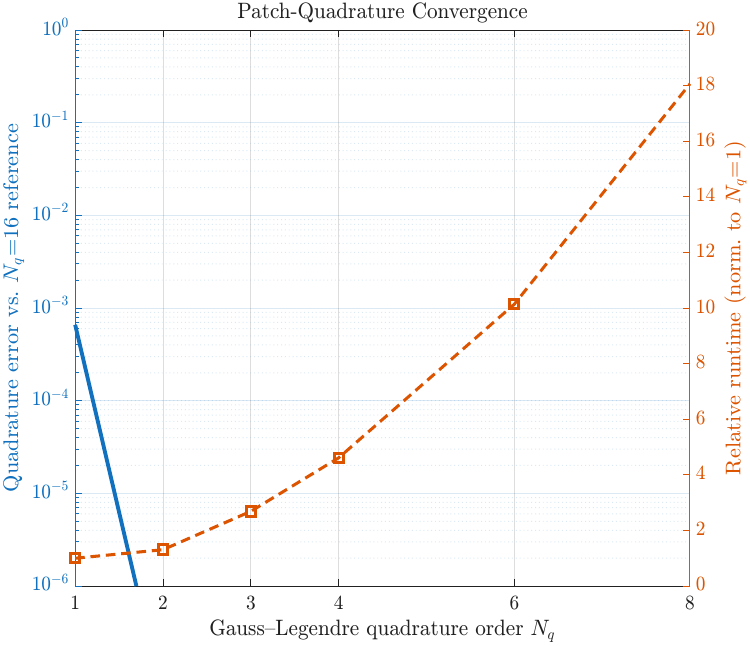}
    \caption{Patch-quadrature convergence against an $N_q=16$ self-reference (left axis) and relative per-patch runtime (right axis). A relative error below $10^{-6}$ is reached already at $N_q=2$, justifying that choice throughout.}\label{fig:quad_conv}
    \vspace{-2ex}
\end{figure}

\begin{figure}[t!]
    \centering
    \includegraphics[width=\columnwidth]{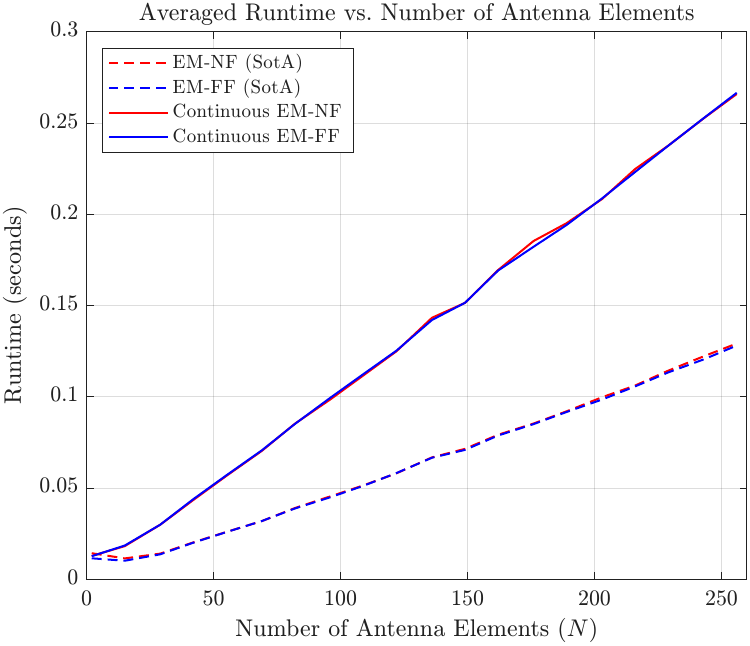}
    \caption{Average runtime comparison of the \ac{SotA} and proposed models.}
    \label{fig:runtimes}
    \vspace{-3ex}
\end{figure}

\begin{figure}[t!]
    \centering
    \includegraphics[width=\columnwidth]{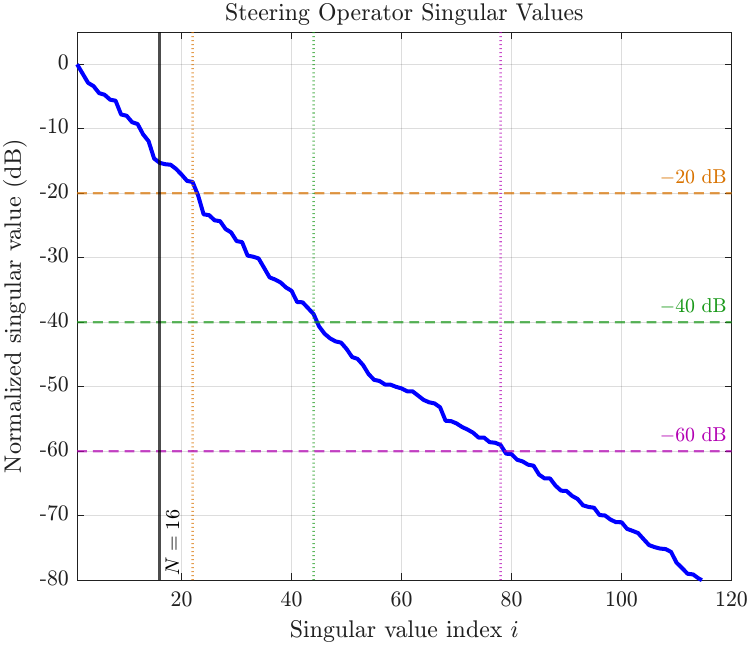}
    \caption{Normalized singular value spectrum of the discretized steering operator $\mathbf{H}_\Omega$ for the $4\times4$ bowtie array at $5\,\mathrm{GHz}$. The $-40\,\mathrm{dB}$ effective rank ($r_{0.01}\!=\!44$) greatly exceeds the $N=16$ hardware ports yet is below the dimension $K=1{,}152$.}
    \label{fig:sv_decay}
    \vspace{2ex}
    \centering
    \includegraphics[width=\columnwidth]{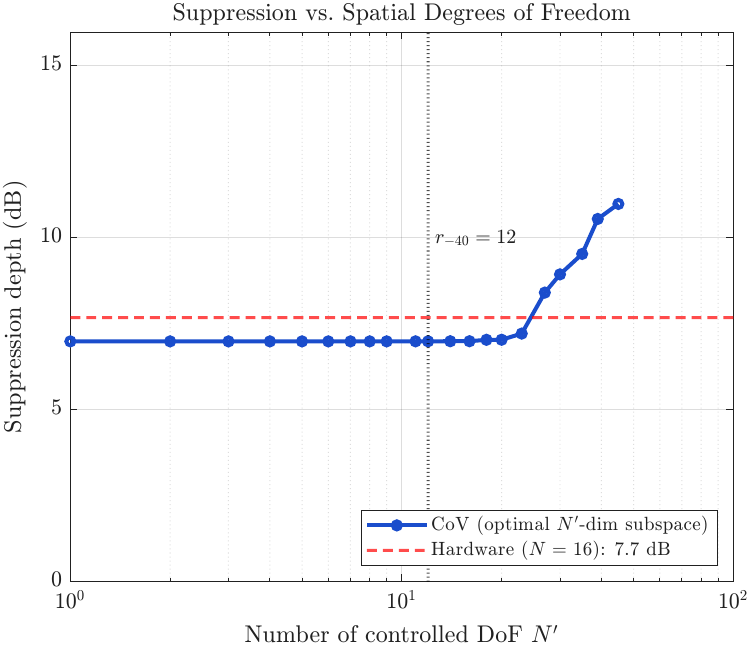}
    \caption{Suppression depth in the region-constrained problem of Section~\ref{subsec:PD_BF} versus the number of controlled degrees of freedom $N'$, using the leading right singular vectors of the suppression-region channel as the control basis. The $N=16$ hardware feed attains $7.7\,\mathrm{dB}$; the continuous model exploits higher-order modes to improve suppression by several dB. The marker $r_{-40}=12$ is the effective rank of the suppression-region channel.}\label{fig:dof_scaling}
    \vspace{-4ex}
\end{figure}

Finally, to ensure that the proposed method does not increase computational complexity, we showcase the average runtime in Fig. \ref{fig:runtimes}.

\subsection{Spectral Characterization and Effective Rank of the Steering Operator}
\label{sec:spectral}

The continuous steering operator $\mathcal{A}: L^2(\mathcal{S}_T) \to L^2(\mathcal{S}_R)$,
\begin{equation}
    [\mathcal{A}w](\mathbf{r}) \triangleq \int_{\mathcal{S}_T} w(\mathbf{p})\,a(\mathbf{r},\mathbf{p})\,\mathrm{d}\mathbf{p},
    \label{eq:steering_op}
\end{equation}
is a Hilbert--Schmidt operator under the regularity conditions of Remark~\ref{rem:regularity},
and therefore admits a countable singular value decomposition
$\mathcal{A} = \sum_{i=1}^\infty \sigma_i \langle \cdot, v_i \rangle_{L^2(\mathcal{S}_T)} u_i$,
with $\sigma_1 \geq \sigma_2 \geq \cdots \geq 0$~\cite{kress1999linear}.

The rate of singular value decay quantifies the effective spatial \ac{DoF} of the aperture.
Rapid decay implies a low-dimensional dominant subspace, whereas slow
decay indicates that a large number of independent feeding modes contribute to the radiated field over the observation domain.

To characterize this numerically, we evaluate the discretized steering matrix
\begin{equation}
    \mathbf{H}_\Omega \triangleq
    \begin{bmatrix} \mathbf{h}_c(\mathbf{r}_1) \\ \vdots \\ \mathbf{h}_c(\mathbf{r}_{N_\Omega}) \end{bmatrix}
    \in \mathbb{C}^{N_\Omega \times K},
    \label{eq:H_obs}
\end{equation}
where $\{\mathbf{r}_i\}_{i=1}^{N_\Omega}$ is a uniform angular grid over the observation sector (azimuth $0^\circ$ to $350^\circ$ over $36$ points and elevation $-40^\circ$ to $40^\circ$ over $9$ points, i.e., $N_\Omega = 324$) at distance $R_\mathrm{eval}$, and
$\mathbf{h}_c(\mathbf{r}) = \mathbf{u}_r\trans \mathbf{K}_\mathrm{cont}(\mathbf{r})\mathbf{M}_c
\in \mathbb{C}^{1 \times K}$ is the scalar channel vector in the $K$-dimensional control space.
The singular values $\{\sigma_i(\mathbf{H}_\Omega)\}_{i=1}^{\min(N_\Omega, K)}$ provide a
finite-dimensional approximation of the operator spectrum.
The \emph{effective numerical rank} is defined as
\begin{equation}
    r_\epsilon \triangleq \left|\left\{i : \sigma_i(\mathbf{H}_\Omega) / \sigma_1(\mathbf{H}_\Omega) \geq \epsilon \right\}\right|,
    \label{eq:eff_rank}
\end{equation}
for a threshold $\epsilon > 0$ (e.g., $\epsilon = 10^{-2}$, corresponding to a $-40\,\mathrm{dB}$ level).

Fig.~\ref{fig:sv_decay} shows the normalized singular value profile of $\mathbf{H}_\Omega$ for
the $4\times 4$ bowtie array at $5\,\mathrm{GHz}$.
The effective rank at the $-40\,\mathrm{dB}$ threshold is $r_{0.01} = 44$, confirming that the dominant radiating subspace of the aperture is substantially lower-dimensional than the $K = 1{,}152$-dimensional continuous control space, yet remains far larger than the $N = 16$ ports of the discrete feed. Consequently, a $16$-port network cannot span this dominant subspace: $r_{0.01} - N = 28$ independent spatial modes supported by the surface geometry are inaccessible to the discrete system. The practical value of these otherwise-unexploited modes is quantified in Fig.~\ref{fig:dof_scaling}, which reports the suppression depth attainable in the region-constrained problem of Section~\ref{subsec:PD_BF} versus the controlled subspace dimension $N'$. The $N=16$ hardware feed attains only $7.7\,\mathrm{dB}$; the continuous model matches this at small $N'$ but, by accessing the higher-order modes, improves the suppression by several decibels as $N'$ grows into the tens. This is a tangible beamforming gain, not merely a nominal increase in control dimension.

\section{Beamforming via Feed Optimization}
\label{sec:bf_cov}

\vspace{-1ex}
\subsection{Beamforming for Maximizing Field Strength}

Let us consider \eqref{eq:scalar_e_field} to formulate the optimization problem
\vspace{-1ex}
\begin{align}
\vspace{-1ex}
\underset{w(\mathbf{p})}{\text{maximize}} \quad &\bigg|\int_{\mathcal{S}_\mathrm{T}} w(\mathbf{p}) a(\mathbf{r}, \mathbf{p}) \, {\rm d}\mathbf{p} \bigg|^2 \label{eq:field_strength_max} \\
\text{s.t.} &\int_{\mathcal{S}_\mathrm{T}} \big| w(\mathbf{p}) \big|^2 \, {\rm d}\mathbf{p} \leq P. \label{eq:power_const}
\end{align}

\newpage

In order to obtain a solution to \eqref{eq:field_strength_max} based on the \ac{CoV}, let us first derive an equivalent power constraint for \eqref{eq:power_const} using the lemma defined below.

\begin{lemma}[\textit{Equivalent Power Constraint}]
\label{lemma:power_const}
The optimal solution to the problem in \eqref{eq:field_strength_max} satisfies the power constraint with equality; i.e.,
\begin{equation}
\label{eq:equiv_power_const}
\int_{\mathcal{S}_\mathrm{T}} \big| w(\mathbf{p}) \big|^2 \, {\rm d}\mathbf{p} = P.
\end{equation}

\begin{proof}
Let $\tilde{w}(\mathbf{p})$ denote a feasible solution to problem \eqref{eq:field_strength_max} that satisfies
\begin{equation}
\label{eq:proof_power_equi}
\tilde{P} \triangleq \int_{\mathcal{S}_\mathrm{T}} \big| \tilde{w}(\mathbf{p}) \big|^2 \, {\rm d}\mathbf{p} < P.
\end{equation}

Then, assuming the objective is continuous and strictly increasing in the power level, and by defining a scaling factor $\rho \triangleq P/\tilde{P}$ and a scaled solution $w(\mathbf{p}) \triangleq \sqrt{\rho}  \, \tilde{w}(\mathbf{p})$, it can be readily shown that the maximum objective in \eqref{eq:field_strength_max} achieved by the scaled solution $w(\mathbf{p})$ must be higher than that achieved by the solution $\tilde{w}(\mathbf{p})$ since $\rho > 1$.

Additionally, it can also be shown that
\begin{equation}
\label{eq:proof_power_equi_shown}
\int_{\mathcal{S}_\mathrm{T}} \big| w(\mathbf{p}) \big|^2 \, {\rm d}\mathbf{p} = \rho \int_{\mathcal{S}_\mathrm{T}} \big| \tilde{w}(\mathbf{p}) \big|^2 \, {\rm d}\mathbf{p} = \rho \tilde{P} = P.
\end{equation}

The results in \eqref{eq:proof_power_equi_shown} implies that for any feasible solution to \eqref{eq:field_strength_max}, there exists a solution that achieves a larger maximum objective with a corresponding power equality constraint. 
The proof is therefore complete.
\end{proof}

\end{lemma}

\begin{figure}[t!]
    \centering
    \includegraphics[width=\columnwidth]{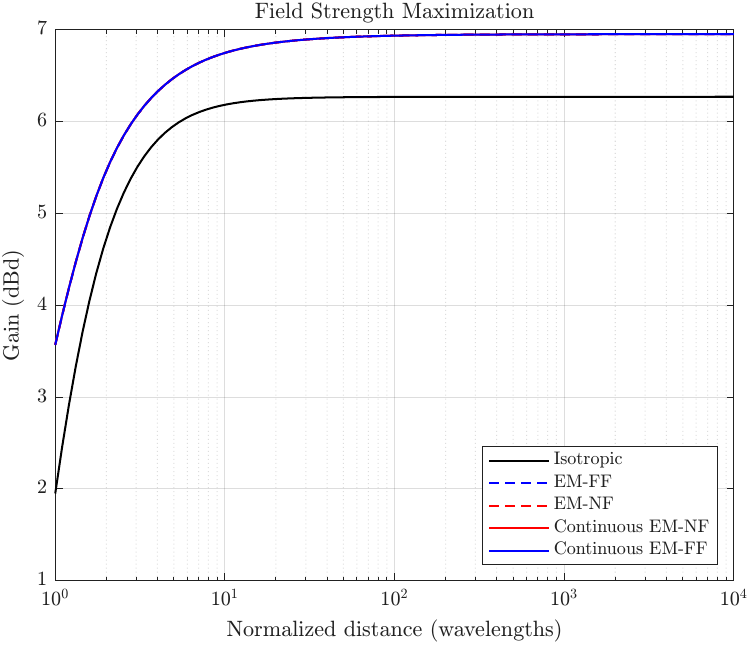}
    \caption{Field-strength-maximization gain (Theorem~\ref{theorem_1}) versus normalized distance for a linear dipole array. The \ac{SotA} (EM-NF/EM-FF) and proposed continuous (Continuous EM-NF/EM-FF) beamformers coincide and converge to the same far-field gain, well above the isotropic reference: single-point focusing is a rank-one problem for which $N=16$ ports already suffice, so the continuous model introduces no penalty.}\label{fig:gain_linear}
\end{figure}

The solution to the problem with the equality constraint is characterized by the following theorem.

\begin{theorem} \label{theorem_1}
    The optimal solution to the problem in equation \eqref{eq:field_strength_max} is given by
    \begin{equation} 
    \label{optimal_solution_SU}
        w^\star(\mathbf{p}) = \sqrt{\frac{P}{\int_{\mathcal{S}_\mathrm{T}} \big| a(\mathbf{r}, \mathbf{p}) \big|^2 \, {\rm d}\mathbf{p}}} \; a^*(\mathbf{r}, \mathbf{p}).
    \end{equation}

\end{theorem}

\begin{IEEEproof}
    Please refer to the Appendix \ref{theorem_1_proof}.
\end{IEEEproof}

\begin{remark}
The result in Theorem~\ref{theorem_1} is the infinite-dimensional analogue of the classical Rayleigh quotient maximization problem
$\max_{\mathbf{w}} \frac{|\mathbf{a}^H\mathbf{w}|^2}{\|\mathbf{w}\|_2^2}$,
whose solution is given by matched filtering \cite{CastellanosTWC2025}.
Here, the finite-dimensional steering vector is replaced by the continuous steering function $a(\mathbf{r},\mathbf{p})$, and the Euclidean norm is replaced by the $L^2(\mathcal{S}_\mathrm{T})$ norm.
\end{remark}

In addition, this solution already takes into account polarization of the potential receiver via the channel model in \eqref{eq:scalar_e_field}, analogous to Section IV.B in \cite{CastellanosTWC2025}.

Fig.~\ref{fig:gain_linear} presents the gain achieved by the field-strength maximization beamformer of Theorem~\ref{theorem_1} as a function of distance, for a linear dipole array. 
As expected, all models converge to the same far-field gain, since single-point focusing is a rank-one problem for which $N = 16$ ports already provide sufficient degrees of freedom. 
This result therefore serves as a \emph{sanity check}: it confirms that the proposed continuous model does not degrade performance in the regime where the \ac{SotA} is already optimal, while the improvements from 
the continuous framework emerge precisely in the overconstrained regional synthesis problem of Section~\ref{subsec:PD_BF}, where $N = 16$ ports are fundamentally insufficient.

\subsection{Generalized Near-Field Pattern Synthesis via Spatial Region Constraints}
\label{subsec:PD_BF}

Modern near-field applications, such as multi-user \ac{MIMO}, physical-layer security, and electromagnetic exposure compliance,  require precise control over radiated fields across extended spatial regions, not merely at isolated points. 
Discrete array models are ill-suited for this task, as their finite degrees of freedom often produce severe 
spatial rippling and leakage when suppressing extended areas. 
To address this, we formulate a region-constrained near-field beamforming problem within the proposed 
continuous \ac{CoV} framework, using a plane-wave power density (PD) constraint as also seen in \cite{CastellanosTWC2025}.

The plane-wave PD at a point $\mathbf{r}$ is given by~\cite{CastellanosTWC2025}
\begin{equation}
    \mathrm{PD}(\mathbf{r}) = \frac{\|\mathbf{e}(\mathbf{r})\|^2}{2\eta_0},
\end{equation}
where $\eta_0$ is the free-space impedance. 

Let $\mathcal{P}_\mathrm{con}$ denote a continuous spatial exclusion region -- representing, for instance, an EMF safety volume or an interference-suppression sector. 
Then, the spatially averaged PD over $\mathcal{P}_\mathrm{con}$ can be expressed as
\begin{align}
    \mathrm{PD}(\mathcal{P}_\mathrm{con})
    &= \frac{1}{2\eta_0 |\mathcal{P}_\mathrm{con}|}
       \int_{\mathcal{P}_\mathrm{con}}
       \bigg\|\int_{\mathcal{S}_\mathrm{T}} w(\mathbf{p})\,\mathbf{a}(\mathbf{r},\mathbf{p})
       \,\mathrm{d}\mathbf{p}\bigg\|^2 \mathrm{d}\mathbf{r} \nonumber\\
    &= \int_{\mathcal{S}_\mathrm{T}}\!\int_{\mathcal{S}_\mathrm{T}}
       w^*(\mathbf{p})\,X_{\mathcal{P}_\mathrm{con}}(\mathbf{p},\mathbf{p}')\,w(\mathbf{p}')
       \,\mathrm{d}\mathbf{p}\,\mathrm{d}\mathbf{p}',
\end{align}
where the spatial suppression kernel is defined as
\begin{equation}
    X_{\mathcal{P}_\mathrm{con}}(\mathbf{p},\mathbf{p}')
    \triangleq
    \frac{1}{2\eta_0|\mathcal{P}_\mathrm{con}|}
    \int_{\mathcal{P}_\mathrm{con}}
    \mathbf{a}^{\mathrm{H}}(\mathbf{r},\mathbf{p})\,\mathbf{a}(\mathbf{r},\mathbf{p}')
    \,\mathrm{d}\mathbf{r}.
\end{equation}

The objective is to maximize the field strength at a target location $\mathbf{r}_u$, subject to a constraint on the average PD over $\mathcal{P}_\mathrm{con}$ given by
\begin{align}
    \underset{w(\mathbf{p})}{\mathrm{maximize}} \quad
    &\bigg|\int_{\mathcal{S}_\mathrm{T}} w(\mathbf{p})\,a(\mathbf{r}_u,\mathbf{p})
    \,\mathrm{d}\mathbf{p}\bigg|^2
    \label{eq:field_strength_max_PD}\\
    \mathrm{s.t.} \quad
    &\int_{\mathcal{S}_\mathrm{T}}\!\int_{\mathcal{S}_\mathrm{T}}
    w^*(\mathbf{p})\,X_{\mathcal{P}_\mathrm{con}}(\mathbf{p},\mathbf{p}')\,w(\mathbf{p}')
    \,\mathrm{d}\mathbf{p}\,\mathrm{d}\mathbf{p}' \leq Q,
    \label{eq:PD_const}
\end{align}
where $Q$ is the maximum permissible average PD within $\mathcal{P}_\mathrm{con}$.

\begin{figure}[t!]
    \centering
    \includegraphics[width=\columnwidth]{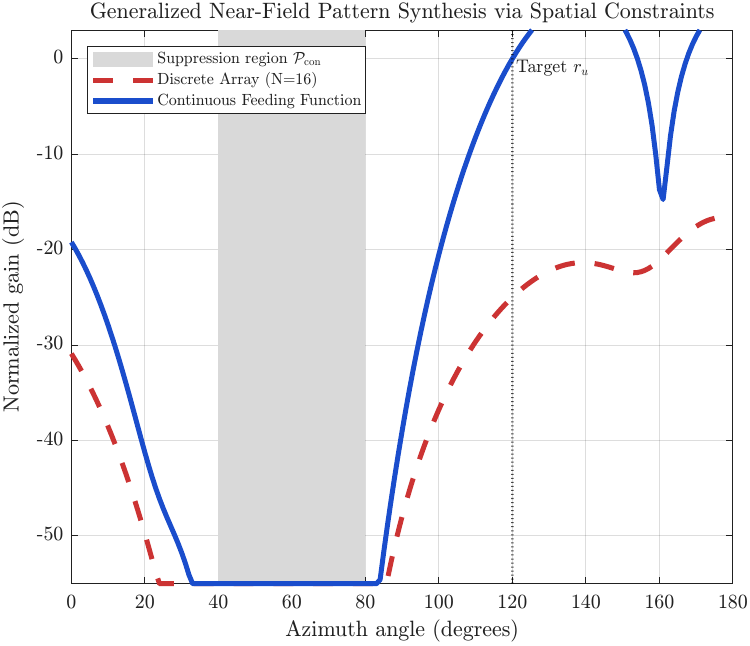}
    \caption{PD-constrained beamforming. Normalized radiation patterns for the discrete ($N=16$) and continuous ($K=1{,}152$) beamformers, both subject to average PD budget $Q = 0.2\,\mathrm{PD}_\mathrm{MF}$ over the suppression region $\mathcal{P}_\mathrm{con}$ (shaded). Both are evaluated on the physical continuous channel model. Target direction: $120^\circ$ azimuth.}
    \label{fig:PD-bF}
\end{figure}

\begin{lemma}[\textit{Active PD Constraint}]\label{lem:active_PD}
The optimal solution to~\eqref{eq:field_strength_max_PD} satisfies the PD 
constraint~\eqref{eq:PD_const} with equality.
\begin{proof}
Since both the objective and the PD functional are degree-two homogeneous in $w$, any 
feasible $\tilde{w}$ satisfying~\eqref{eq:PD_const} with strict inequality can be scaled 
by $\rho = \sqrt{Q \big/ \int\!\int \tilde{w}^* X_{\mathcal{P}_\mathrm{con}} \tilde{w}} > 1$, 
yielding a strictly larger objective that still satisfies~\eqref{eq:PD_const} with equality. 
Hence the constraint is active at the optimum.
\end{proof}
\end{lemma}

\begin{theorem}\label{thm:PD_opt}
Assume $X_{\mathcal{P}_\mathrm{con}}(\mathbf{p},\mathbf{p}')$ induces a strictly positive 
definite integral operator $\mathcal{X}_{\mathcal{P}_\mathrm{con}}$ on 
$L^2(\mathcal{S}_\mathrm{T})$. The optimal solution to~\eqref{eq:field_strength_max_PD} is
\begin{equation}\label{eq:w_opt_PD}
    w^\star(\mathbf{p}) = \sqrt{\frac{Q}{\Lambda}}\,
    \bigl[\mathcal{X}^{-1}_{\mathcal{P}_\mathrm{con}}\,a^*(\mathbf{r}_u,\cdot)\bigr](\mathbf{p}),
\end{equation}
where $[\mathcal{X}^{-1}_{\mathcal{P}_\mathrm{con}} f](\mathbf{p})
\triangleq \int_{\mathcal{S}_\mathrm{T}} X^{-1}_{\mathcal{P}_\mathrm{con}}(\mathbf{p},\mathbf{p}')\,
f(\mathbf{p}')\,\mathrm{d}\mathbf{p}'$ denotes the inverse operator applied to $f$, and
\begin{equation}\label{eq:Lambda_def}
    \Lambda \triangleq
    \int_{\mathcal{S}_\mathrm{T}}\!\int_{\mathcal{S}_\mathrm{T}}
    a(\mathbf{r}_u,\mathbf{p})\,X^{-1}_{\mathcal{P}_\mathrm{con}}(\mathbf{p},\mathbf{p}')\,
    a^*(\mathbf{r}_u,\mathbf{p}')\,\mathrm{d}\mathbf{p}\,\mathrm{d}\mathbf{p}'.
\end{equation}
\end{theorem}
\begin{proof}
    Please refer to Appendix~\ref{theorem_2_proof}.
\end{proof}

\begin{remark}
Theorem~\ref{thm:PD_opt} is the continuous analogue of the discrete generalized matched 
filter $\mathbf{w}^\star \propto \mathbf{X}^{-1}\mathbf{a}$. The matrix $\mathbf{X}$ is 
replaced by the integral operator $\mathcal{X}_{\mathcal{P}_\mathrm{con}}$, the steering 
vector $\mathbf{a}$ by $a^*(\mathbf{r}_u,\cdot) \in L^2(\mathcal{S}_\mathrm{T})$, and the 
scalar $\mathbf{a}^H\mathbf{X}^{-1}\mathbf{a}$ by $\Lambda$ in~\eqref{eq:Lambda_def}. 
Note that $\Lambda = \|\mathcal{X}^{-1/2}_{\mathcal{P}_\mathrm{con}} a^*(\mathbf{r}_u,\cdot)\|^2_{L^2} > 0$
whenever $a(\mathbf{r}_u,\cdot)\not\equiv 0$ on $\mathcal{S}_\mathrm{T}$. As in
Theorem~1, the integrals in~\eqref{eq:w_opt_PD}--\eqref{eq:Lambda_def} admit
efficient evaluation via GL quadrature.
\end{remark}

\begin{remark}[Ill-Posedness and Regularization]
\label{rem:illposed}
The strict positive-definiteness assumed in Theorem~\ref{thm:PD_opt} is a working hypothesis rather than a generic property. Since $\mathcal{X}_{\mathcal{P}_\mathrm{con}}$ is a region-averaged Gram operator built from the compact steering operator, it is itself compact: its singular values accumulate at zero, so $0$ belongs to its spectrum, the inverse $\mathcal{X}^{-1}_{\mathcal{P}_\mathrm{con}}$ is unbounded, and the Fredholm equation of the first kind~\eqref{eq:Fredholm_PD} is ill-posed in the sense of Hadamard. Physically, the exclusion region $\mathcal{P}_\mathrm{con}$ constrains only a finite-dimensional subspace of feeding functions with dimension equal to the effective rank of the region channel ($r_{-40}=12$ for the $\mathcal{P}_\mathrm{con}$ of Fig.~\ref{fig:dof_scaling}) while leaving the complementary directions unconstrained. Theorem~\ref{thm:PD_opt} should therefore be interpreted as the exact minimum-norm solution on the range of $\mathcal{X}_{\mathcal{P}_\mathrm{con}}$, obtained in practice as the $\varepsilon\!\to\!0^+$ limit of the Tikhonov-regularized operator $\mathcal{X}_{\mathcal{P}_\mathrm{con}}+\varepsilon\mathcal{I}$; equivalently, $\mathcal{X}^{-1}_{\mathcal{P}_\mathrm{con}}$ is read as the Moore--Penrose pseudo-inverse. This is the exact continuous counterpart of the rank-deficient discrete generalized matched filter, where $\mathbf{X}^{-1}$ is likewise replaced by a regularized (diagonally-loaded) inverse.
\end{remark}

Fig.~\ref{fig:PD-bF} compares the normalized radiation patterns of the discrete $N = 16$ port beamformer and the proposed continuous ($K = 1{,}152$) beamformer, both designed via Theorem~\ref{thm:PD_opt} with PD budget $Q = 0.2 \cdot \mathrm{PD}_\mathrm{MF}$, where $\mathrm{PD}_\mathrm{MF}$ is the average PD in $\mathcal{P}_\mathrm{con}$ under uncontrolled matched-filter beamforming. 
The suppression region spans azimuth $[40^\circ, 80^\circ]$ at elevation $30^\circ$ and distances $\{1.0, 1.5, 2.0\}\lambda$.

\begin{figure}[t!]
    \centering
    \includegraphics[width=\columnwidth]{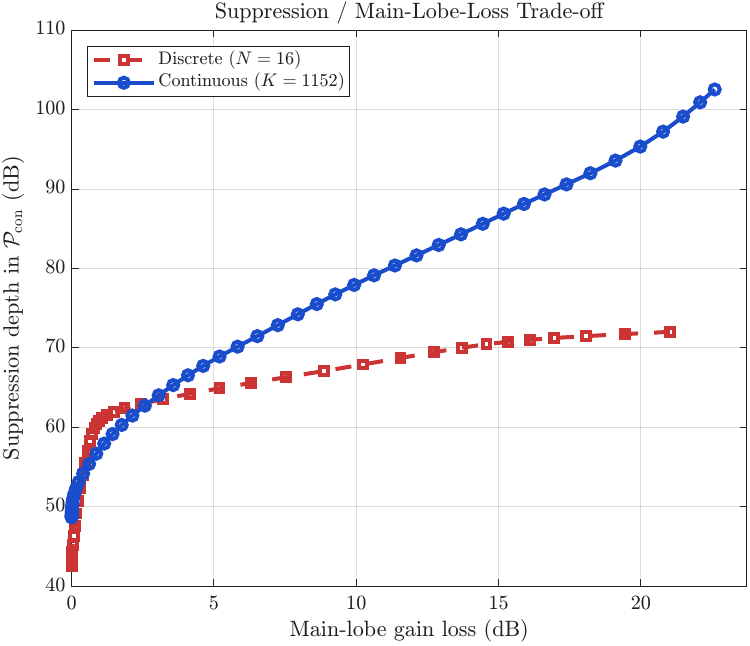}
    \caption{Pattern suppression in $\mathcal{P}_\mathrm{con}$ (region power density below the main-lobe gain) versus main-lobe gain loss, traced by a diagonal-loading sweep, with both beamformers on the physical continuous channel. The discrete ($N=16$) array saturates near $72\,\mathrm{dB}$ once its degrees of freedom are exhausted, whereas the continuous ($K=1{,}152$) design keeps deepening the null beyond $100\,\mathrm{dB}$; the curves cross near $2.5\,\mathrm{dB}$ of main-lobe loss.}\label{fig:tradeoff}
\end{figure}

Two suppression metrics are presented here and should not be conflated. In Figs.~\ref{fig:PD-bF} and~\ref{fig:tradeoff}, \emph{pattern suppression} is the average power density within $\mathcal{P}_\mathrm{con}$ relative to the main-lobe (target) gain, $10\log_{10}(G/\overline{\mathrm{PD}}_{\mathcal{P}_\mathrm{con}})$. Since it is a ratio of physical field quantities, it is directly comparable across the discrete and continuous parameterizations. (In Fig.~\ref{fig:PD-bF} the patterns are clamped at $-55\,\mathrm{dB}$ for legibility, so the depths quoted below lie beneath the visible floor.) In Fig.~\ref{fig:dof_scaling}, by contrast, \emph{suppression depth} is referenced to the unconstrained matched filter, $-10\log_{10}(\overline{\mathrm{PD}}_{\mathcal{P}_\mathrm{con}}/\mathrm{PD}_\mathrm{MF})$, and is on the order of $7$--$11\,\mathrm{dB}$ under the budget $Q = 0.2\,\mathrm{PD}_\mathrm{MF}$, hence the different scale.

This comparison reflects deployment as it would actually occur: the discrete beamformer is designed under the point-source model $\mathbf{K}_\mathrm{disc}$ of~\cite{CastellanosTWC2025} (the \ac{SotA} design assumption), the proposed beamformer under the accurate patch model $\mathbf{K}_\mathrm{cont}$, and both are then evaluated on the true physical channel.
In terms of the pattern-suppression metric, the discrete beamformer achieves an average suppression of approximately $74.7$~dB within $\mathcal{P}_\mathrm{con}$, while the proposed continuous beamformer achieves $77$~dB suppression, meaning a gain of $2.3$~dB, at a main-lobe loss of less than $1$~dB relative to unconstrained matched filtering.
Quantitatively, the \emph{actual} PD incurred by the discrete beamformer in $\mathcal{P}_\mathrm{con}$ is
$\mathrm{PD}_\mathrm{disc} / Q \approx 82.5\%$ of the budget, while the continuous beamformer achieves exactly $100\%$ utilization of $Q$, confirming Lemma~\ref{lem:active_PD}.
At this shallow operating point the advantage stems chiefly from the proposed model's near-field accuracy: the budget shortfall quantified above is a direct symptom of the discrete beamformer's point-source design mismatch on the true channel, from which the accurately-modeled continuous framework is free. The complementary, and ultimately larger, benefit of the expanded spatial degrees of freedom emerges in the deep-suppression regime of Fig.~\ref{fig:tradeoff}.

The single operating point of Fig.~\ref{fig:PD-bF} generalizes to the full trade-off in Fig.~\ref{fig:tradeoff}, obtained by sweeping a diagonal-loading parameter $\mu$ in the regularized beamformer $w(\mu)\propto(\mathcal{X}_{\mathcal{P}_\mathrm{con}}+\mu\mathcal{I})^{-1}a^*(\mathbf{r}_u,\cdot)$, so that each point trades main-lobe gain for pattern suppression. In contrast to Fig.~\ref{fig:PD-bF}, here the discrete beamformer is \emph{also} granted the accurate model $\mathbf{K}_\mathrm{cont}$: both designs and evaluations use the true physical channel, so the comparison removes the modeling-accuracy effect of Fig.~\ref{fig:PD-bF} and isolates the influence of the degrees of freedom alone. The two are essentially equivalent for shallow suppression, where the discrete array's sixteen full-vector degrees of freedom already suffice. As deeper suppression is demanded, however, their behavior diverges sharply: the discrete pattern \emph{saturates} near $72\,\mathrm{dB}$. This is because its finite spatial degrees of freedom are exhausted, and no further main-lobe sacrifice buys additional suppression, whereas the continuous beamformer continues to deepen the null beyond $100\,\mathrm{dB}$, the two curves crossing near $2.5\,\mathrm{dB}$ of main-lobe loss. In the deep-suppression regime demanded by electromagnetic-exposure compliance and strong interference nulling, the continuous framework is therefore decisively superior, precisely because it is not bound by the port count.
The absolute suppression depths shown are those of the idealized noiseless model with perfect channel knowledge and would, in practice, be floored by hardware imperfections and channel-estimation error. The robust and physically meaningful conclusion is thus the \emph{qualitative} insight that the finite-port array saturates whereas the continuous aperture does not, rather than the specific decibel values.

\section{Conclusion}
\label{sec:conclusion}

We have presented a unified \ac{CoV} framework for the characterization of, and beamforming over, continuous \ac{EM} manifolds of arbitrary \ac{MIMO} array geometries, simultaneously resolving the three principal limitations of the discrete \ac{SotA}.
First, a patch-based \ac{GL} radiation operator was shown to replace the point-source centroid approximation of~\cite{CastellanosTWC2025}, delivering consistent near-field accuracy improvements across every configuration tested, including linear and planar, dipole and bowtie, half-wavelength and wide-spaced, at a computational cost that differs only by a constant factor independent of the array size.
Second, a continuous feeding function $w(\mathbf{p})\in L^2(\mathcal{S}_\mathrm{T})$, introduced as the infinite-dimensional limit of the $N$-port network, was shown to lift the beamforming space from the $N$-dimensional port subspace onto a hardware-decoupled $K$-dimensional subspace, with a spectral analysis of the steering operator quantifying the additional radiation modes thereby unlocked.
Third, exploiting this continuous formulation, we derived closed-form optimal beamformers via the \ac{CoV} for both field-strength maximization and region-constrained near-field pattern synthesis under a \ac{PD} constraint, and established their exact analogy to the discrete matched and generalized matched filters.
Numerical results confirmed that the continuous beamformer attains markedly sharper spatial suppression than its discrete counterpart precisely in the over-constrained regional-synthesis regime where finite-port arrays are fundamentally deficient, while incurring no penalty in the rank-one focusing regime where the \ac{SotA} is already optimal.
Future work will address the design of structured $N$-port feed networks whose aggregate response approaches the continuous aperture bound, the extension to multi-user and volumetric geometries, and the incorporation of the framework into near-field sensing and \ac{ISAC} pipelines.

\begin{appendices}

    \section{Proof of Theorem \ref{theorem_1}} \label{theorem_1_proof}

    This appendix provides the proof for the optimal continuous beamformer stated in Theorem \ref{theorem_1}. 
    The proof is derived using the calculus of variations based on the \ac{KKT} conditions. 
    The Lagrangian function of problem \eqref{eq:field_strength_max} is given by 
    \begin{align}
        \mathcal{L}(w) = & \bigg|\int_{\mathcal{S}_\mathrm{T}} w(\mathbf{p}) a(\mathbf{r}, \mathbf{p}) \, {\rm d}\mathbf{p} \bigg|^2 \!\!- \mu \bigg( \int_{\mathcal{S}_\mathrm{T}} \big| w(\mathbf{p}) \big|^2 \, {\rm d}\mathbf{p} - P \bigg),
    \end{align}
    where $\mu \ge 0$ is the Lagrange multiplier for the power constraint. 
    
    Following the principles of the calculus of variations, the optimal $w(\mathbf{p})$ maximizing $\mathcal{L}(w)$ is found where the first variation of $\mathcal{L}(w)$, denoted by $\delta \mathcal{L}(w, \delta w)$ is zero for any perturbation $\delta w$. 
    The first variation can be expressed as
    \begin{align}
        \delta \mathcal{L}(w, \delta w) & \!=\! \left. \frac{d}{d \epsilon} \mathcal{L}(w + \epsilon \delta w) \right|_{\epsilon = 0}\!\!\! = 2\Re\!\left\{ \int_{\mathcal{S}_\mathrm{T}} \!\!\delta w^*(\mathbf{p}) \chi(\mathbf{p}) {\rm d}\mathbf{p}\! \right\}\!,
    \end{align} 
    where
    \begin{equation}
         \chi(\mathbf{p}) \triangleq \bigg(\int_{\mathcal{S}_\mathrm{T}} w(\mathbf{z}) a(\mathbf{r}, \mathbf{z}) \, {\rm d}\mathbf{z}\bigg) a^*(\mathbf{r}, \mathbf{p}) - \mu w(\mathbf{p}).
    \end{equation}
    
    For the functional $\mathcal{L}(w)$ to be at a maximum, we must have $\delta \mathcal{L}(w, \delta w) = 0$ for any arbitrary $\delta w$, implying $\chi(\mathbf{p})$ must be zero, i.e., 
    \begin{equation}
        \mu\, w(\mathbf{p}) = a^*(\mathbf{r}, \mathbf{p}) \int_{\mathcal{S}_\mathrm{T}} w(\mathbf{z}) a(\mathbf{r}, \mathbf{z}) \, {\rm d}\mathbf{z}.
    \end{equation}

    This can be written in Fredholm form as
    \begin{equation}
    \mu\,w(\mathbf{p})
    =
    \int_{\mathcal{S}_\mathrm{T}}
    K(\mathbf{p},\mathbf{z})\,w(\mathbf{z})\,{\rm d}\mathbf{z},
    \end{equation}
    where the kernel is defined as
    \begin{equation}
    K(\mathbf{p},\mathbf{z})
    \triangleq
    a^*(\mathbf{r},\mathbf{p})\,a(\mathbf{r},\mathbf{z}) \in \mathbb{C}.
    \end{equation}

    Since $K(\mathbf{p},\mathbf{z})$ is separable and rank-one, the corresponding eigenspace is one-dimensional, implying
    \begin{equation}
    w^\star(\mathbf{p}) \propto a^*(\mathbf{r},\mathbf{p}).
    \end{equation}
    
    Finally, enforcing the power constraint
    \begin{equation}
    \int_{\mathcal{S}_\mathrm{T}} |w(\mathbf{p})|^2{\rm d}\mathbf{p} = P,
    \end{equation}
    yields the normalized optimal solution
    \begin{equation}
    w^\star(\mathbf{p})
    =
    \sqrt{\frac{P}{\int_{\mathcal{S}_\mathrm{T}} \big| a(\mathbf{r}, \mathbf{p}) \big|^2 \, {\rm d}\mathbf{p}}} \; a^*(\mathbf{r}, \mathbf{p}).
    \end{equation}

    This completes the proof.

    \section{Proof of Theorem~\ref{thm:PD_opt}} \label{theorem_2_proof}

    By Lemma~\ref{lem:active_PD}, the inequality constraint~\eqref{eq:PD_const} is active at the
    optimum, so we solve the equality-constrained problem. 
    The Lagrangian can then be expressed as
    \begin{align}
        \mathcal{L}(w)
        &= \biggl|\int_{\mathcal{S}_\mathrm{T}} w(\mathbf{p})\,a(\mathbf{r}_u,\mathbf{p})\,\mathrm{d}\mathbf{p}\biggr|^2
        \nonumber \\
        &\hspace{-3ex}- \mu \biggl(
          \int_{\mathcal{S}_\mathrm{T}}\!\int_{\mathcal{S}_\mathrm{T}}
          w^*(\mathbf{p})\,X_{\mathcal{P}_\mathrm{con}}(\mathbf{p},\mathbf{p}')\,w(\mathbf{p}')\,
          \mathrm{d}\mathbf{p}\,\mathrm{d}\mathbf{p}' - Q
        \biggr),
    \end{align}
    where $\mu \geq 0$ is the Lagrange multiplier. 
    
    Defining $c \triangleq \int_{\mathcal{S}_\mathrm{T}} w(\mathbf{p})\,a(\mathbf{r}_u,\mathbf{p})\,\mathrm{d}\mathbf{p} \in \mathbb{C}$,
    and following the \ac{CoV} approach of Appendix \ref{theorem_1_proof}, the first variation with respect to $w^*$ is given by
    \begin{equation}
        \delta\mathcal{L}(w,\delta w) = 2\,\mathrm{Re}
        \int_{\mathcal{S}_\mathrm{T}} \delta w^*(\mathbf{p})\,\chi(\mathbf{p})\,\mathrm{d}\mathbf{p},
    \end{equation}
    where
    \begin{equation}
        \chi(\mathbf{p}) \triangleq
        c\,a^*(\mathbf{r}_u,\mathbf{p})
        - \mu \int_{\mathcal{S}_\mathrm{T}} X_{\mathcal{P}_\mathrm{con}}(\mathbf{p},\mathbf{p}')\,w(\mathbf{p}')\,\mathrm{d}\mathbf{p}'.
    \end{equation}
    
    Requiring $\delta\mathcal{L} = 0$ for all perturbations $\delta w$ implies $\chi(\mathbf{p}) = 0$,
    which yields the Fredholm integral equation of the first kind
    \begin{equation}\label{eq:Fredholm_PD}
        \bigl(\mathcal{X}_{\mathcal{P}_\mathrm{con}}\,w\bigr)(\mathbf{p})
        = \frac{c}{\mu}\,a^*(\mathbf{r}_u,\mathbf{p}).
    \end{equation}
    
    Since $\mathcal{X}_{\mathcal{P}_\mathrm{con}}$ is strictly positive definite and hence invertible (or, in the rank-deficient case, invertible on its range in the regularized sense of Remark~\ref{rem:illposed}),
    equation~\eqref{eq:Fredholm_PD} has the unique solution
    \begin{equation}\label{eq:w_unnorm_PD}
        w(\mathbf{p}) = \frac{c}{\mu}\,
        \bigl[\mathcal{X}^{-1}_{\mathcal{P}_\mathrm{con}}\,a^*(\mathbf{r}_u,\cdot)\bigr](\mathbf{p}),
    \end{equation}
    confirming $w^\star(\mathbf{p}) \propto [\mathcal{X}^{-1}_{\mathcal{P}_\mathrm{con}}\,a^*(\mathbf{r}_u,\cdot)](\mathbf{p})$.
    
    \textit{Normalization.}
    Substituting~\eqref{eq:w_unnorm_PD} into the active PD constraint and using the self-adjointness of $\mathcal{X}_{\mathcal{P}_\mathrm{con}}$ yields
    \begin{align}
        Q &\!=\! \frac{|c|^2}{\mu^2} \!\!
        \int_{\mathcal{S}_\mathrm{T}}\!\int_{\mathcal{S}_\mathrm{T}}
        \!\!\!\!\bigl[\mathcal{X}^{-1} \! a^*\bigr]\!^*\!(\mathbf{p})
        X_{\mathcal{P}_\mathrm{con}}(\mathbf{p},\mathbf{p}')
        \bigl[\mathcal{X}^{-1} \! a\!^*\!\bigr]\!(\mathbf{p}')\,
        \mathrm{d}\mathbf{p}\,\mathrm{d}\mathbf{p}'
        \nonumber \\
        &= \frac{|c|^2}{\mu^2}\,
        \bigl\langle
          \mathcal{X}^{-1}_{\mathcal{P}_\mathrm{con}}\,a^*,\;
          \mathcal{X}_{\mathcal{P}_\mathrm{con}}\,\mathcal{X}^{-1}_{\mathcal{P}_\mathrm{con}}\,a^*
        \bigr\rangle_{L^2}
        \nonumber \\
        &= \frac{|c|^2}{\mu^2}\,
        \bigl\langle
          \mathcal{X}^{-1}_{\mathcal{P}_\mathrm{con}}\,a^*,\; a^*
        \bigr\rangle_{L^2},
        \label{eq:norm_step}
    \end{align}
    where $\langle f, g \rangle_{L^2} \triangleq \int_{\mathcal{S}_\mathrm{T}} f^*(\mathbf{p})\,g(\mathbf{p})\,\mathrm{d}\mathbf{p}$.
    
    Expanding the inner product in~\eqref{eq:norm_step} using the Hermitian symmetry of the
    inverse kernel, $\overline{X^{-1}_{\mathcal{P}_\mathrm{con}}(\mathbf{p},\mathbf{p}')}
    = X^{-1}_{\mathcal{P}_\mathrm{con}}(\mathbf{p}',\mathbf{p})$, and relabelling the integration
    variables, one obtains
    \begin{align}
        &\bigl\langle \mathcal{X}^{-1}_{\mathcal{P}_\mathrm{con}}\,a^*,\;a^*\bigr\rangle_{L^2}
        \nonumber \\
        &= \int_{\mathcal{S}_\mathrm{T}}\!\int_{\mathcal{S}_\mathrm{T}}
          a(\mathbf{r}_u,\mathbf{p})\,X^{-1}_{\mathcal{P}_\mathrm{con}}(\mathbf{p},\mathbf{p}')\,
          a^*(\mathbf{r}_u,\mathbf{p}')\,\mathrm{d}\mathbf{p}\,\mathrm{d}\mathbf{p}'
        = \Lambda.
    \end{align}
    
    Hence $|c|/\mu = \sqrt{Q/\Lambda}$, and substituting back into~\eqref{eq:w_unnorm_PD}
    yields~\eqref{eq:w_opt_PD}. This completes the proof. \hfill$\blacksquare$

\end{appendices}

\bibliographystyle{IEEEtran}
\bibliography{references}

\end{document}